\documentclass[conference]{IEEEtran}
\IEEEoverridecommandlockouts

\usepackage{cite}
\usepackage{amsmath,amssymb,amsfonts}
\usepackage{bm,mathtools,comment}
\usepackage{graphicx}
\usepackage{xcolor}
\usepackage{url}
\usepackage{subcaption}
\usepackage{amsthm}
\newtheorem{proposition}{Proposition}

\DeclareMathOperator*{\argmin}{arg\,min}

\DeclareMathOperator{\tr}{tr}
\DeclareMathOperator{\vecop}{vec}
\DeclareMathOperator{\diag}{diag}

\def\RR{\mathbb{R}}

\def\cN{\mathcal{N}}
\def\bI{\mathbf{I}}
\def\be{\mathbf{e}}
\def\by{\mathbf{y}}
\def\bY{\mathbf{Y}}
\def\bu{\mathbf{u}}

\def\bA{\mathbf{A}}

\def\bK{\mathbf{K}}

\def\br{\mathbf{r}}

\def\bs{\mathbf{s}}

\def\bSigma{\mathbf{\Sigma}}

\begin{document}

\title{Causal Spatio-Temporal Sound Field Reconstruction
}

\author{\IEEEauthorblockN{David Sundstr\"om, Filip Tronarp, Johan Lindstr\"om, Andreas Jakobsson}
\IEEEauthorblockA{Centre for Mathematical Sciences, Lund University, Sweden}
}

\maketitle

\begin{abstract}
In sound field control applications, it is commonly assumed that one has access to an accurate representation of the sound field in the region of interest. 
This is a problematic assumption since the reconstruction of a sound field from available microphone measurements is especially challenging in real-time applications where only causal measurements are available. 
Notably, causal time-windowed observations introduce correlation between frequency components, making sound field reconstruction methods that process each frequency band independently sub-optimal.  
In this work, we formulate a causal finite-window spatio-temporal linear minimum mean-square error estimator for sound field reconstruction.
The sound field is modeled as the solution to the wave equation driven by a stationary stochastic spatio-temporal source distribution, which induces a physically interpretable covariance function. 
It is shown that this covariance function is closely related to the classical diffuse-field coherence model.
Since the computational complexity grows rapidly with the number of spatio-temporal observations, we formulate a budget-constrained spatio-temporal sample selection approach to minimize the posterior reconstruction variance.
The proposed estimator and sampling strategy are evaluated using both simulated and measured sound fields, demonstrating improved short-window reconstruction compared to frequency domain finite-window baselines. 
\end{abstract}

\section{Introduction}
\label{sect:background}

\IEEEPARstart{R}{econstructing} a sound field from spatially distributed microphone measurements is central in applications where the acoustic pressure must be known, controlled, or rendered over an extended region.
Examples include spatial active noise control \cite{zhang2018active}, sound zone control \cite{brunnstrom2022variable}, and spatial audio reproduction \cite{amengual2021optimizations}. 
The idea of controlling a sound field in locations away from physical microphone locations was demonstrated  early on \cite{garcia1997generation,spors2007approach}.
More recently, interest has emerged in controlling or reproducing a sound field in a continuous target region, typically relying on the assumed availability of an accurate reconstruction of the sound field \cite{zhang2018active,Koyama2021,brunnstrom2022variable,amengual2021optimizations}. 
However, the problem of reconstructing the sound field from microphone measurements is non-trivial and remains challenging in many applications where real-time constraints are imposed. 
In these settings, the sound field must be reconstructed using only causal observations and typically with limited computational resources, which is the focus here. %in this paper.

A large body of work on sound field reconstruction is based on physical dictionary-based representations of acoustic fields. Classical approaches include modal and spherical harmonic representations, plane wave decompositions, and equivalent source methods \cite{williams1999fourier,koopmann1989method,johnson1998equivalent, samarasinghe2015efficient}. Related sparse and regularized formulations have been used for the interpolation of sound fields in room acoustic environments, often by exploiting equivalent source models or plane wave structures \cite{tervo2013spatial,mignot2013room,verburg2018reconstruction, antonello2017room,hahmann2022convolutional,sundstrom2023recursive}. 
These methods impose the acoustic structure in a physically meaningful way and can provide accurate reconstructions from limited spatial measurements. 
However, the models typically rely on estimating a large set of parameters and are thus suitable for offline settings rather than addressing the here considered causal finite-window estimation problem. 

A closely related class of reconstruction methods is based on kernel ridge regression and Gaussian-process (GP) regression. In the setting with Gaussian observation noise, the kernel ridge regression estimator can be interpreted as the posterior mean of a Gaussian process with the kernel as covariance function \cite{wahba1990spline,rasmussen2005gaussian}. We therefore refer to this literature jointly as kernel/GP sound field reconstruction. 
This viewpoint is attractive because physical prior knowledge can be encoded in the covariance function while retaining a closed form linear estimator and, in the GP interpretation, a posterior measure of uncertainty. Several sound field reconstruction methods have exploited kernels constrained by the Helmholtz equation, including diffuse field kernels \cite{ueno2017sound,ueno2018kernel}, directionally weighted kernels \cite{horiuchi2021kernel}, and learned kernel parameterizations \cite{ribeiro2024sound,sundstrom2024sound}. We refer to \cite{caviedes2021gaussian} for an overview of GP models related to sound field reconstruction. The resulting linear estimators have also been used directly for sound field control formulations \cite{Koyama2021,brunnstrom2022variable}. 
Many of these estimators are formulated per frequency, 
%posed frequency by frequency, 
which is natural for stationary fields observed over sufficiently long time intervals. 
However, this does not specify estimators associated with a short causal observation window.

The finite window setting is important because a rectangular observation window generally induces spectral leakage and cross-correlation between discrete Fourier coefficients. Consequently, estimators that process short-window frequency bins independently neglect part of the covariance structure available in the observations. This issue is closely related to the role of finite data windows, transients, and periodicity assumptions in frequency domain system identification \cite{pintelon2002frequency,schoukens2004time,ljung2004state}. Thus, the limitation is not the frequency domain representation of the sound field itself. In fact, a frequency domain estimator retaining the full cross-frequency covariance would be equivalent to a corresponding time-domain estimator. 
%Rather, the approximation arises when the causal finite-window covariance is replaced by independent frequency-bin problems.
Rather, the approximation arises when the causal finite-window covariance is assumed to be independent between frequency bins.

The limitation of kernel-based models that are defined for single frequencies was recently discussed in the context of sound field estimation in \cite{brunnstrom2025time}. 
Additionally, a framework was proposed to build discrete time sound field kernel models whose frequency components must be considered jointly. 
It was furthermore shown that this can be useful for jointly regularizing the spatial interpolation of room impulse responses (RIRs) with temporal, spatial, and spectral structures. 
Although our goal is different, the estimator proposed in this paper may be interpreted within such a framework by choosing a kernel that corresponds to the spatio-temporal covariance induced by the stochastic wave equation model derived below. 

Spatio-temporal sound field models have also been proposed in the context of RIR reconstruction. Examples include sparse spatio-temporal models \cite{antonello2017room}, Bayesian and GP formulations \cite{fernandez2021reconstruction,caviedes2023spatio}, optimal transport-based methods \cite{Sundstrom2023}, and physics-informed neural-network approaches \cite{karakonstantis2024room,olivieri2024physics}.
These works are important since they model the information contained in the joint spatial and temporal structure of acoustic fields. 
Their typical objective, however, is offline interpolation or extrapolation of RIRs over extended domains, often using rich nonstationary covariance models, learned representations, or iterative inference procedures.
In contrast, the present work %instead 
addresses a setting related to low-latency reconstruction of the current sound pressure field from a short causal window of streaming microphone measurements.

To obtain a physically interpretable covariance appropriate for sound field control applications, i.e., 
not requiring an
%without 
online estimation of a large number of spatial parameters, we model the sound field as the solution to the wave equation driven by a stationary stochastic source distribution. The proposed estimator is closely related to classical diffuse-field models \cite{cook1955measurement,jacobsen2013fundamentals,badeau2024statistical} and may be interpreted as a time domain generalization of the classical diffuse-field covariance. 
Unlike frequency domain bin-wise diffuse kernels applied to a short observation window, the proposed covariance retains the temporal correlations induced by the finite causal window.

Finally, the causal spatio-temporal formulation has a computational complexity that scales rapidly with the number of spatio-temporal samples in the observation window, a limiting %practical 
aspect in practical %for these 
applications. 
To alleviate this, we here
%We therefore 
formulate the problem of selecting the (fixed) number of spatio-temporal samples that minimizes the reconstruction variance in the target region. 
This formulation may be viewed as a spatio-temporal generalization of the optimal sensor selection problems that have been studied in various forms, for example based on mutual information \cite{ariga2020mutual}, conditional entropy \cite{nishida2022region}, or using a Bayesian formulation  \cite{verburg2024optimal}.
The present work assumes a fixed microphone geometry and selects an informative subset of the spatio-temporal samples available in a causal window. This turns the computational budget into a sampling design problem over both the microphone index and the time lag.

The contributions of this work are threefold. First, we formulate causal finite-window sound field reconstruction as a spatio-temporal linear minimum mean-square error (LMMSE) problem using a covariance induced by a stochastic wave equation model. Second, we show that the proposed covariance reduces to the classical diffuse-field coherence model in the far-field limit while retaining the finite-window temporal correlations needed for causal reconstruction. Third, we introduce a budget-constrained spatio-temporal sample selection problem, which reduces the number of observations used by the estimator by minimizing the posterior reconstruction variance.

The remainder of the paper is organized as follows. Section~\ref{sec:problem_statement} formulates the causal finite-window sound field reconstruction problem and introduces the observation model. Section~\ref{sec:spde_measurement} derives the stochastic wave-equation covariance model and relates it to classical diffuse-field coherence models. Section~\ref{sec:reconstruction} presents the resulting spatio-temporal linear MMSE estimator. Section~\ref{sec:optimal_sampling} formulates budget-constrained spatio-temporal sample selection for reducing computational complexity. Section~\ref{sec:num_experiments} evaluates the proposed reconstruction and sampling methods on simulated and measured sound fields. Finally, Section~\ref{sec:conclusion} concludes the paper.

%=========================================================
\section{Problem Statement}
\def\cZ{\mathbb{Z}}
\def\brhat{\hat{\mathbf{r}}}
\label{sec:problem_statement}
We represent the sound field as a spatio-temporal stochastic process whose distribution is generated by the wave equation
\begin{align}
  \frac{1}{c^2}\frac{\partial^2}{\partial t^2}u(t,\br)
  - \Delta u(t,\br)
  = s(t,\br),
  \label{eq:wave_eq}
\end{align}
where $(t,\br)\in\RR\times\RR^3$, $c\in \RR_+$ is the speed of sound, and $s(t,\br)$ is a stochastic source distribution. 
There are two main reasons for studying the problem using a stochastic source distribution. First, it allows for the introduction of uncertainty into the model for settings where a deterministic source distribution is not fully known. Second, several acoustic phenomena, such as late reverberation in room acoustics \cite{badeau2024statistical}, are well described as stochastic fields. In practice, the sound field in \eqref{eq:wave_eq} is measured under additive noise for $W\in \cZ_+$ time samples using $M\in \cZ_+$ spatially distributed microphones, resulting in the measurement model
\begin{align}
  y_m[n] = u(t_n,\br_m) + v_m[n],
  \label{eq:y_meas}
\end{align}
where $\{\br_m\}_{m=1}^M$ are the microphone positions and $v_m[n]$ is the observation noise that is assumed to be well modeled as being independent and identically distributed Gaussian noise with variance $\sigma^2$, i.e., $v_m[n]\sim \cN(0,\sigma^2)$. 
For notational simplicity, we 
assume uniform sampling, \(t_n=nT_s\), \(T_s=1/f_s\).
Let the stacked observation matrix be defined as
\begin{align}
  \by_0[n]
  &=
  \begin{bmatrix}y_1[n] & \dots & y_M[n]\end{bmatrix}^\top\in\RR^M,
  \nonumber\\
  \bY[n]
  &=
  \begin{bmatrix}\by_0[n] & \dots & \by_0[n-W+1]\end{bmatrix}\in\RR^{M\times W}.
  \label{eq:Y_stack}
\end{align}
Below, in Section~\ref{sec:reconstruction}, we consider the problem of estimating the sound field 
\begin{align}
  \bu[n]
  &=
  \begin{bmatrix}
    u(t_n,\brhat_1) & \dots & u(t_n,\brhat_P)
  \end{bmatrix}^\top\in\RR^P,
  \label{eq:U_stack}
\end{align}
given the observations $\bY[n]$, i.e., the estimation 
of the sound field at time $t_n$ at $P$ locations, $\{\brhat_p\}_{p=1}^P$, using a causal finite window of length $W$ from each of the $M$ sensors. 
As the resulting computational complexity grows rapidly with both $W$ and $M$, we will in the following also %Section~\ref{sec:optimal_sampling} 
examine how an informative subset of $K$ samples may be selected from this candidate set, substantially reducing the computational cost of the resulting estimator.

%=========================================================
\section{Stochastic Wave Equation Measurement Model}
\label{sec:spde_measurement}
We proceed to characterize the joint spatio-temporal distribution of the sound field under the model in \eqref{eq:wave_eq}. We assume that \(s\) is a zero-mean GP with covariance as specified below in Section~\ref{sec:source_model}. Since the wave equation is linear, the induced sound field \(u\) is also a GP, whose mean and covariance are derived in the following.
The sound field in \eqref{eq:wave_eq} may be represented in terms of its Green's function such that
\begin{align}
  u(t,\br)
  = \int_{\RR^3}\int_{\RR}G(t,\br;\tau,\br')\,s(\tau,\br')\,d\tau\,d\br',
  \label{eq:u_green}
\end{align}
where the spatio-temporal Green's function is defined as the solution to
\begin{align}
  \frac{1}{c^2}\frac{\partial^2}{\partial t^2}G(t,\br;\tau,\br') - \Delta G(t,\br;\tau,\br')
  &= \delta(t-\tau)\delta(\br-\br').
  \label{eq:green_def}
\end{align}
In 3D 
free space, the Green's function is given by
\begin{align}
  G(t,\br;\tau,\br')
  &= \frac{1}{4\pi \|\br-\br'\|_2}\,
    \delta\!\Big(t-\tau-\frac{\|\br-\br'\|_2}{c}\Big).
  \label{eq:green_3d}
\end{align}
The expectation and covariance of \eqref{eq:u_green} can thus be written as
\begin{align}
  \mathbb{E}[u(t,\br)]
  &=
  \int_{\RR^3}\int_{\RR}G(t,\br;\tau,\br')\,\mathbb{E}[s(\tau,\br')]\,d\tau\,d\br'
  \label{eq:u_mean}
\end{align}
and 
\begin{align}
  &C_u\!\big((t,\br),(t',\br')\big)
  =  \nonumber \\
  &\int_{\RR}\!\int_{\RR^3}\!\int_{\RR}\!\int_{\RR^3}\!
  G(t,\br;\tau,\br_0)\,G(t',\br';\tau',\br_1)C_s((\tau,\br_0),(\tau',\br_1)) \nonumber \\ 
  & \hspace{15em} \times d\tau\,d\br_0\,d\tau'\,d\br_1, 
  \label{eq:u_covariance}
\end{align}
respectively, where $C_s$ denotes the spatio-temporal covariance of the source distribution. Note in particular that the distribution of the sound field can be characterized by the mean and covariance of the source distribution. 
However, the four integrals over unbounded domains in \eqref{eq:u_covariance} make numerical calculations of the covariance impractical without imposing further assumptions on the source distribution.

\subsection{Source distribution model}
\def\GP{\mathcal{GP}}
\def\wt{\tilde{w}}
\label{sec:source_model}
Herein, we model the source distribution as a zero-mean GP,
\begin{align}
s \sim \mathcal{GP}(0,C_s),
\end{align}
with separable covariance
\begin{align}
 C_s\!\big((t,\br),(t',\br')\big)
 =
 q\,\tilde w(\br,\br')\kappa(t-t'),
  \label{eq:source_cov}
\end{align}
where \(q\in\RR_+\) is the source intensity, \(\tilde w: \RR^3\times\RR^3\to\RR\)  the spatial covariance, and
\(\kappa:\RR\to\RR\)  the stationary temporal covariance.
The assumption of temporal stationarity should be interpreted as a local approximation over the short (causal) windows used here. This is consistent with standard short-time modeling of audio and speech signals, where signals are often assumed to be approximately stationary over a short window \cite{Naylor2011}.

In realistic acoustic scenes, the spatial statistics may be both spatially structured and time-varying.
Although various specialized parameterizations have been proposed, including \cite{fernandez2021reconstruction} and \cite{sundstrom2025bayesian}, the online estimation of parameters in the covariance function typically introduces high computational complexity and thus falls outside the intended applications considered in this work (for the interested reader, we refer to the literature on offline estimation techniques \cite{sundstrom2023recursive,brunnstrom2025spatial,fernandez2021reconstruction, sundstrom2024sound}). In fact, applications with limited computational resources may not afford the computations required for the joint online estimation of the directivity parameters in the covariance, and must resort to simpler covariance options despite the introduction of a model misspecification. For example, the use of a frequency domain diffuse covariance function was recently shown to be useful for spatial active noise control in \cite{Koyama2021}. Therefore, we focus on the important special setting of a diffuse covariance function by imposing similar assumptions on the source distribution. However, it is worth noting that the methodology developed in this paper holds for any known covariance structure of the form \eqref{eq:source_cov}. The exposition is therefore restricted to a covariance structure of the form
\def\SS{\mathcal{S}}
\begin{align}
 C_s\!\big((t,\br),(t',\br')\big)
 = 
 q\,w(\br)w(\br')\,\delta_{\SS_a}(\br-\br')\,\kappa(t-t'), 
 \label{eq:source_cov_2}
\end{align}
where $\br,\br'\in \SS_a$,
i.e., it is assumed that the source is spatially white on the bounded domain of the surface of a sphere of radius $a$, denoted as $\SS_a$. Here, \(\delta_{\SS_a}\) denotes the Dirac delta with respect to the surface measure on \(\SS_a\).
The radius \(a\) is chosen such that the source sphere encloses the reconstruction and microphone regions. However, as shown in Section~\ref{sec:num_experiments}, the reconstruction is insensitive to the choice of \(a\).
Furthermore, we will here assume $w(\br) = 1$, such that the sound field corresponding to \eqref{eq:source_cov_2} is diffuse with a uniform directivity pattern. 
Evaluating the covariance function in \eqref{eq:u_covariance} using the assumed structure of the source distribution in \eqref{eq:source_cov_2} yields (see Appendix~\ref{app:u_covariance_derivation})
\begin{align}
  C_u\!\big((t,\br),&(t',\br')\big) = \frac{q}{16\pi^2}
  \int_{\SS_a}
  \frac{1}{\|\br-\br_0\|_2\|\br'-\br_0\|_2} \nonumber\\
  & \hspace{-5mm} \times \kappa\!\Big(
    (t-t') - \frac{\|\br-\br_0\|_2-\|\br'-\br_0\|_2}{c}
  \Big)\,dS(\br_0)
  \label{eq:Cu_integral}
\end{align}
allowing 
the integrals in \eqref{eq:u_covariance} to be %have been
reduced to a more numerically tractable integral over $\SS_a$. 
Furthermore, to define $\kappa$ we for simplicity and to isolate the effect of finite-window spatio-temporal reconstruction assume a flat source spectrum over the frequency band of interest, with a spectral density given by
\begin{align}
  \Phi(\omega)
  =
  \begin{cases}
    \pi/(\omega_2-\omega_1), & \omega\in[\omega_1,\omega_2]\cup[-\omega_2,-\omega_1],\\
    0, & \text{otherwise},
  \end{cases}
  \label{eq:Phi_bp}
\end{align}
where $\omega = 2\pi f$.
The induced stationary covariance is the inverse Fourier transform, given by
\begin{align}
  \kappa(\Delta)
  &= \frac{1}{2\pi}\int_{\RR}\Phi(\omega)e^{j\omega\Delta}\,d\omega
  \nonumber\\
  &= \frac{\sin(\omega_2\Delta)-\sin(\omega_1\Delta)}{\Delta (\omega_2 - \omega_1)},
  \label{eq:kappa_bp}
\end{align}
with \(\kappa(0)=1\).
%
% in \eqref{eq:Cu_integral}. 
It is worth noting that the covariance in \eqref{eq:Cu_integral} is stationary in time and can, for a discrete time setting, be written as 
\begin{align}
  R_u(\br,\br';\tau)
  &=
  C_u\!\big((0,\br),(-\tau,\br')\big) \nonumber \\
  &=
  \mathbb E[u(t,\br)u(t-\tau,\br')],
  \\
  C(\br,\br';\ell)
  &=
  R_u(\br,\br';\ell T_s),
  \qquad \ell\in\mathbb Z.
  \label{eq:C_lag}
\end{align}
In this work, we approximate the covariance using a simple equal-area quadrature on the sphere surface given by
\begin{align}
  C(\br,\br';\ell)
  &\approx
  \frac{q}{16\pi^2}
  \sum_{i=1}^Q
  \frac{1}{\|\br-\br_i\|\,\|\br'-\br_i\|} \nonumber\\
  &\times\kappa\!\Big(\ell T_s - \tfrac{\|\br-\br_i\|-\|\br'-\br_i\|}{c}\Big)
  \frac{4\pi a^2}{Q},
  \label{eq:C_quadrature}
\end{align}
where $\br_i$, for \(i=1,\ldots,Q\), are distributed on a nearly uniform Fibonacci lattice on the surface of the sphere \cite{gonzalez2010measurement}. 
To sum up the assumptions introduced in this section, it has in \eqref{eq:source_cov} been assumed that the source distribution is stationary in time, in \eqref{eq:source_cov_2} that it is  spatially white on $\SS_a$, and with a flat spectrum in \eqref{eq:Phi_bp}. 
This yields a covariance with no unknown spatial directivity parameters.
The proposed set of assumptions corresponds to the important special case of a diffuse sound field. 
To make this connection explicit, we next relate the proposed covariance model to established diffuse-field models in the frequency domain.

\subsection{Connection to established diffuse covariance models in the frequency domain}
\label{sec:freq_domain_connection}

To relate the covariance model in \eqref{eq:Cu_integral} to the established
frequency domain diffuse field correlation function \cite{cook1955measurement}, define the cross-spectral density to the stationary covariance such that
\begin{align}
  S_u(\br,\br';\omega)
  &=
  \int_{\RR} R_u(\br,\br';\tau)\,e^{-j\omega\tau}\,d\tau.
  \label{eq:Su_def}
\end{align}
The following proposition shows that the cross-spectral density of \eqref{eq:Su_def} reduces to the classical diffuse-field model in \cite{cook1955measurement} in the far-field limit. Let \(S_{u,a}(\br,\br';\omega)\) denote the cross-spectral density induced by a source sphere of radius \(a\).
\begin{proposition}
\label{prop:diffuse_freq_limit}
Let \(d=\|\br-\br'\|_2\). For fixed interior points \(\br,\br'\),
\begin{align}
  S_{u,a}(\br,\br';\omega)
  \xrightarrow[a\to\infty]{}
  \frac{q\Phi(\omega)}{4\pi}
  \frac{\sin((\omega/c)d)}{(\omega/c)d}.
\end{align}
\end{proposition}
\begin{proof}
See Appendix~\ref{app:freq_domain_connection}.
\end{proof}
In particular, in the far-field limit the corresponding normalized spatial coherence is
\begin{align}
  \gamma(\br,\br';\omega)
  &=
  \frac{S_u(\br,\br';\omega)}
  {\sqrt{S_u(\br,\br;\omega)\,S_u(\br',\br';\omega)}}
  =
  \frac{\sin\!\left(\frac{\omega}{c}\|\br-\br'\|_2\right)}
       {\frac{\omega}{c}\|\br-\br'\|_2},
  \label{eq:gamma_diffuse_main}
\end{align}
which coincides with the classical diffuse-field coherence model in \cite{cook1955measurement}. Thus, the proposed SPDE-induced covariance can be viewed as a finite-radius, time-domain generalization of the classical diffuse-field kernel.

Although this means that the covariance of the sound field model is related to the one used in, for example, \cite{ueno2018kernel}, the resulting kernel estimator when independently processing each frequency bin is not equivalent to the spatio-temporal estimator formulated below. To see this, note that finite windowing of a stationary signal replaces ideal frequency samples by weighted averages of the underlying spectrum, producing spectral leakage and cross-bin covariance. 
A motivating derivation is given in Appendix~\ref{app:windowed_csd_appendix}.  Figure~\ref{fig:windowed_csd_ism_pair12} illustrates how this aspect becomes critical for small $W$.

\begin{figure*}[t]
    \centering
    \begin{subfigure}[t]{0.32\textwidth}
        \centering
        \includegraphics[width=\linewidth]{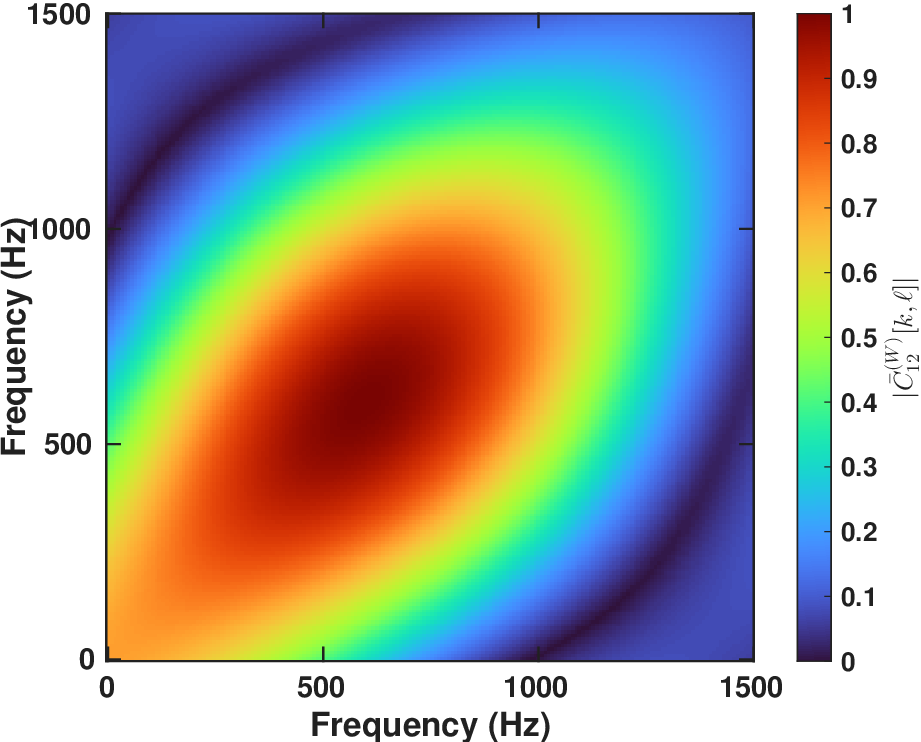}
        \caption{$W=10$.}
        \label{fig:windowed_csd_ism_pair12_w10}
    \end{subfigure}
    \hfill
    \begin{subfigure}[t]{0.32\textwidth}
        \centering
        \includegraphics[width=\linewidth]{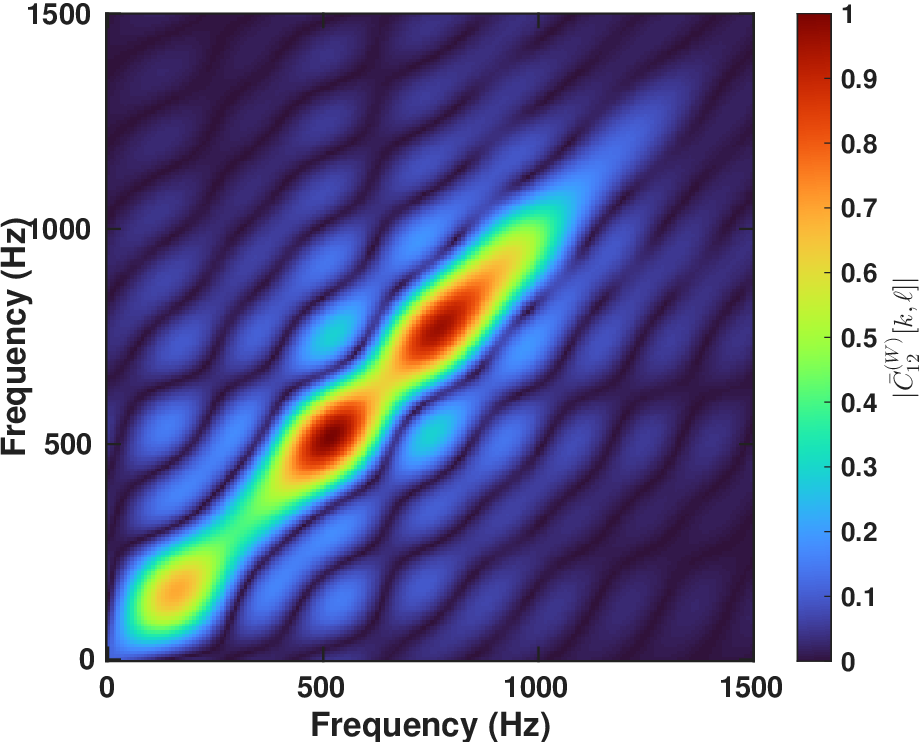}
        \caption{$W=50$.}
        \label{fig:windowed_csd_ism_pair12_w50}
    \end{subfigure}
    \hfill
    \begin{subfigure}[t]{0.32\textwidth}
        \centering
        \includegraphics[width=\linewidth]{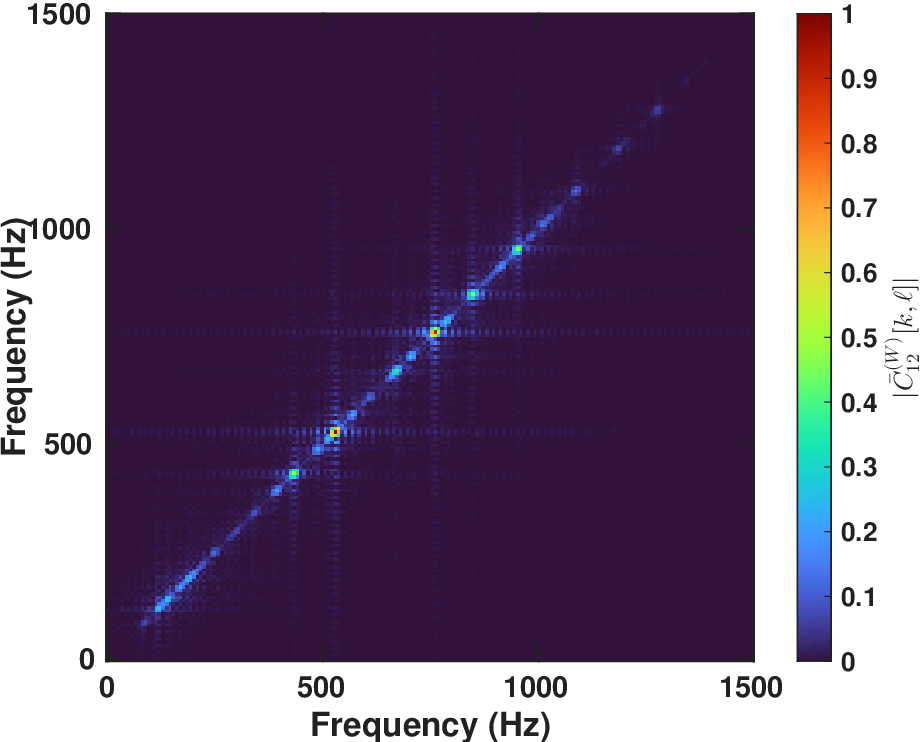}
        \caption{$W=500$.}
        \label{fig:windowed_csd_ism_pair12_w500}
    \end{subfigure}
    \caption{
    Normalized magnitude of the finite-window cross-frequency covariance between microphones \(m=1\) and \(m'=2\) in the simulated image-source room. Specifically, the plotted quantity is
    \(|\bar C_{12}^{(W)}[k,\ell]|\), where
    \(C_{12}^{(W)}[k,\ell]=\mathbb E[U_{1,k}[n]U_{2,\ell}[n]^*]\)
    and \(\bar C_{12}^{(W)}\) is the max-normalized version.
    For short windows, substantial off-diagonal energy is present, showing that the independent frequency approximation is inaccurate.
    }
    \label{fig:windowed_csd_ism_pair12}
\end{figure*}

\section{Sound Field Reconstruction}
\label{sec:reconstruction}

\def\bmuU{\boldsymbol{\mu}_u}
\def\bmuY{\boldsymbol{\mu}_y}
\def\bmuYW{\boldsymbol{\mu}_{y,W}}
\def\buPred{\hat{\bu}^{-}}
\def\byPred{\hat{\by}^{-}}
\def\byPredW{\hat{\by}^{-}_W}
\def\bA{\mathbf{A}}
\def\bR{\mathbf{R}}
Here, the reconstruction of the sound field is formulated as the maximization of the {\em posterior} distribution $p(\bu[n]\mid\bY[n])$. 
Under the Gaussian source model, the stacked prediction and observation vectors are jointly Gaussian. Hence, the maximum {\em a posteriori} estimator, posterior mean, and LMMSE estimator coincide.
Specifically, the joint distribution of the measurements and the sound field in the reconstruction positions is given by \cite{rasmussen2005gaussian}
\begin{align}
  \begin{bmatrix}\bu[n]\\ \by[n]\end{bmatrix}
  &\sim
  \cN\!\left(
    \begin{bmatrix}\mathbf{0}\\ \mathbf{0}\end{bmatrix},\;
    \begin{bmatrix}
      \bK_{uu} & \bK_{uy}\\
      \bK_{yu} & \bK_{yy}+\sigma^2\bI_{MW}
    \end{bmatrix}
  \right),
  \label{eq:joint_gaussian}
\end{align}
where $\bK_{yu}=\bK_{uy}^\top$,
\begin{align}
  \by[n] = \vecop(\bY[n])\in\RR^{MW},
  \label{eq:vec_defs}
\end{align}
with $\vecop(\cdot)$ denoting the vectorization operator, and
\begin{align}
  [\bK_{uu}]_{p,p'}
  &=
  C(\hat{\br}_p,\hat{\br}_{p'};0),
  \\
  [\bK_{yy}]_{(m,w),(m',w')}
  &=
  C(\br_m,\br_{m'};w'-w),
  \label{eq:Kyy}
  \\
  [\bK_{uy}]_{p,(m,w)}
  &=
  C(\hat{\br}_p,\br_m;w),
  \label{eq:Kuy}
\end{align}
where \(m,m'=1,\ldots,M\) and \(w,w'=0,\ldots,W-1\).
Furthermore, the Gaussian posterior is given by
\begin{align}
    p(\bu[n]\mid\by[n])=\cN\!\big(\bu[n];\,\hat{\bu},\,\bSigma_{u\mid y}\big),
  \label{eq:posterior_density}
\end{align}
where the posterior mean and covariance are
\begin{align}
  \hat{\bu}
  &= \bK_{uy}\,\big(\bK_{yy}+\sigma^2\bI_{MW}\big)^{-1}\,\by[n],  \label{eq:post_mean} \\
  \bSigma_{u|y}
  &=
  \bK_{uu}
  - \bK_{uy}
    \Big(\bK_{yy} + \sigma^2\bI_{MW}
    \Big)^{-1}\bK_{yu},
  \label{eq:post_var}
\end{align}
respectively. 
For fixed geometry and covariance parameters, the matrix factorization and the reconstruction filter can be precomputed. Nevertheless, the precomputation cost scales as \(O((MW)^3)\), the memory scales as \(O((MW)^2)\), and the online filtering cost scales as \(O(PMW)\) per time sample. Reducing the number of retained spatio-temporal samples from \(MW\) to \(K\), where $K \ll MW$, therefore notably reduces both offline factorization cost and online filtering cost. 
We will in the following section examine how this can be done.

%=========================================================
\section{Optimal Spatio-Temporal Sampling}
\label{sec:optimal_sampling}
\def\bZ{\mathbf{Z}}
% Show how the measurement model changes with the z. 
We proceed to consider the problem of selecting the $K$ spatio-temporal samples (out of the measured $MW$) that minimizes the reconstruction variance in \eqref{eq:post_var}. In practice, the resulting selection scheme could be pre-computed offline for a fixed geometry and target region to reduce the online filtering cost. The ideas presented here are related to the optimal sensor placement problem considered in \cite{verburg2024optimal}, but generalized to the spatio-temporal setting.

Let
$\mathbf{z}\in\{0,1\}^{MW}$ denote the masking vector, where $z_i=1$ means that the
$i$th spatio-temporal sample is used. The masked measurements are thus given by 
\begin{equation}
    \tilde{\by}[n] = \bZ \by[n],
\end{equation}
where $\bZ=\diag(\mathbf{z})$. We introduce a fixed budget of $K$ spatio-temporal samples, typically determined by the allowed computational complexity of the system, such that 
\begin{align}
  \mathbf{1}^\top \mathbf{z} = K.
  \label{eq:sampling_budget}
\end{align}
The posterior covariance of the reconstructed sound field given the measurements $\tilde{\by}[n]$ is similarly to \eqref{eq:post_var} given by
\begin{align}
  \bSigma_{u|\tilde{y}}(\mathbf{z})
  =
  \bK_{uu}
  - \bK_{uy}\mathbf{Z}
    \Big(
      \mathbf{Z}\bK_{yy}\mathbf{Z} + \sigma^2\bI_{MW}
    \Big)^{-1}
    \mathbf{Z}\bK_{yu}.
  \label{eq:sampling_posterior_cov}
\end{align}
The optimal spatio-temporal selection is formulated as the minimization of the sum of the posterior reconstruction variance, given by 
\begin{align}
  \mathbf{z}^\star
  \in
  \argmin_{\mathbf{z}\in\{0,1\}^{MW}}
  \ &\tr\!\big(\bSigma_{u|\tilde{y}}(\mathbf{z})\big)
  \label{eq:sampling_problem}\\
  &\text{s.t. }
  \mathbf{1}^\top \mathbf{z} = K. \nonumber
\end{align}
However, evaluating the cost for all $\binom{MW}{K}$ combinations of $K$ selections among the set of $MW$ samples is typically computationally intractable. Instead, inspired by the ideas in \cite{verburg2024optimal}, we first solve a relaxed problem to obtain a pruned subset of the $MW$ samples for which a greedy search is computationally tractable.   
The relaxed problem is obtained by relaxing the binary selection vector by $0\le z_i\le 1$.
%In the case where 
When $z_i\geq\varepsilon$, for all $i=1,\dots,MW$ and $\varepsilon>0$, the identity
\[
\mathbf Z
(\mathbf Z\mathbf K_{yy}\mathbf Z+\sigma^2\mathbf I)^{-1}
\mathbf Z
=
(\mathbf K_{yy}+\sigma^2\mathbf Z^{-2})^{-1},
\]
allows the masked posterior covariance to be written as an equivalent heteroscedastic-noise design problem,  where the lower bound \(\varepsilon\) ensures that \(\mathbf Z^{-2}\) is well defined. It may be noted that a small $z_i$ corresponds to a large effective noise variance, suggesting the relaxed design problem on the set
\begin{align}
  \mathcal{H}_K
  =
  \left\{
    \mathbf{z}\in\RR^{MW} :
    \mathbf{1}^\top\mathbf{z}=K,\;
    \varepsilon\le z_i\le 1
  \right\},
  \label{eq:hypersimplex_sampling}
\end{align}
where the objective in \eqref{eq:sampling_problem} is given on the convenient form
\begin{align}
  \phi(\mathbf{z})
  =
  \tr\!\big(\bK_{uu}\big)
  -
  \tr\!\Big(
    \big(
      \bK_{yy} + \sigma^2\mathbf{Z}^{-2}
    \big)^{-1}
    \bK_{yu}\bK_{uy}
  \Big).
  \label{eq:sampling_relaxed_objective}
\end{align}
This allows the relaxed problem to be formulated as
\begin{align}
  \min_{\mathbf{z}\in\RR^{MW}}
  \ & \phi(\mathbf{z}) \quad
  %\\
  \text{s.t. }
  %& 
  \mathbf{z}\in\mathcal{H}_K.
  \label{eq:sampling_relaxed_problem}
\end{align}
which may be efficiently minimized using projected-gradient steps with a backtracking line search. 
The derivation of the gradient and the projection onto $\mathcal{H}_K$ is given in Appendix~\ref{app:projected_gradient}. 

Finally, the relaxed solution $\hat{\mathbf{z}}$ is generally not binary, so it is used as an
importance score for a reduced search. 
We keep the candidates associated with the $\lceil\rho K\rceil$ largest candidates in \(\hat{\mathbf z}\) with $\rho\geq1$, and then run a forward greedy minimization of \eqref{eq:sampling_problem} on this reduced set until $K$ indices have been chosen.
The resulting greedy stage still optimizes the exact posterior-variance criterion, while
the relaxed stage prunes the search space. Note that the greedy search can be implemented using computationally efficient low-rank updates similar to the updates described in \cite{verburg2024optimal}.

%=========================================================
\section{Numerical experiments}
\label{sec:num_experiments}
We evaluate the proposed reconstruction and sample-selection methods using both simulated and measured sound fields. The overall aim of these experiments is to understand the effect of the causal windowed observations, understand the robustness of the proposed method with respect to parameter settings, and to understand the computational gains from the optimal sampling scheme. The excitation signal is band-limited white Gaussian noise in the range \(70\)--\(1000\) Hz. For the image source simulations, we convolve the excitation signal with RIRs simulated using the implementation in \cite{habets2006room}, with room dimensions $3.0\times4.0\times2.5$ m, reflection parameter $0.5$, sampling rate $f_s=8$ kHz, and speed of sound $c=343$ m/s. The microphone array is circular with $M=8$ microphones, radius $0.10$ m, and center $(1.5,1.3,1.2)$ m. The reconstruction region is a circular planar region centered at the array center, with radius $0.05$ m and grid spacing $0.01$ m, corresponding to $P=81$ validation positions. The source position is randomly sampled within the room for each Monte Carlo realization of the experiment. We also use a diffuse free-field simulation with the same microphone and reconstruction geometry. In this simulation, independent realizations of the excitation signals are assigned to $1000$ uniformly distributed directions on a sphere of radius $5$ m, and propagated using the free-field Green's function. If nothing else is mentioned, additive white noise is added to the measurements to achieve an SNR of $20$  dB.

\begin{figure}[t]
    \centering
    \includegraphics[width=\linewidth]{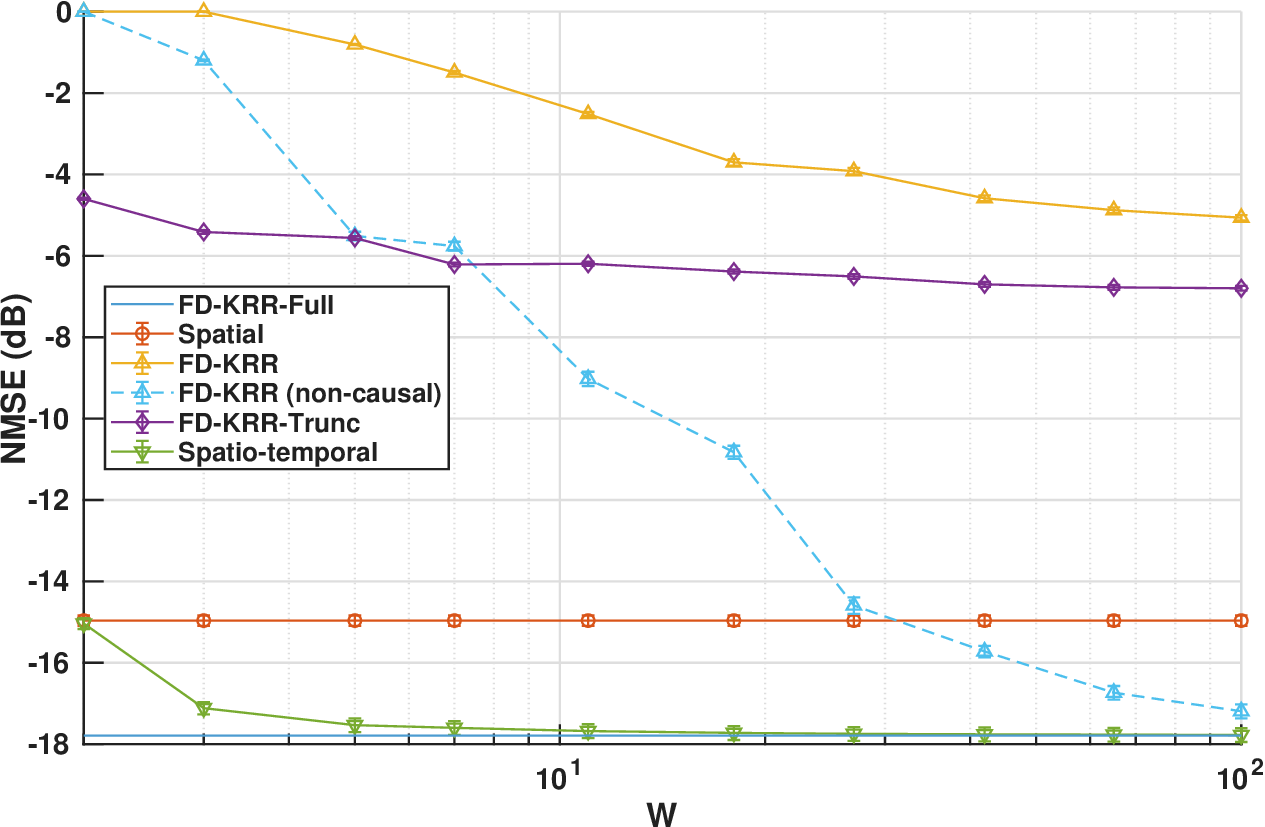}
    \caption{Illustration of the NMSE as a function of the window length $W$ for simulated diffuse-field data. The non-causal finite-window FD-KRR method use windows of length $2W-1$, while the causal methods use length $W$. }
    \label{fig:nmse_vs_winlen_diffuse}
\end{figure}

The measured experiments use the same excitation signal, but with measured RIRs from the DTU dataset %presented in 
\cite{fernandez2021reconstruction}, using the microphones available from the spherical array~1 as observations and the microphones from the linear array~1 as validation points (see Figure~\ref{fig:measured_geometry_dtu019} for an illustration).
The measured RIRs are downsampled to $8$ kHz. For each Monte Carlo realization, M=50 microphones among the 310 spherical-array microphone channels are sampled uniformly without replacement, while all $P=152$ measured linear array positions are used as validation positions.

\begin{figure}[t]
    \centering
    \includegraphics[width=\linewidth]{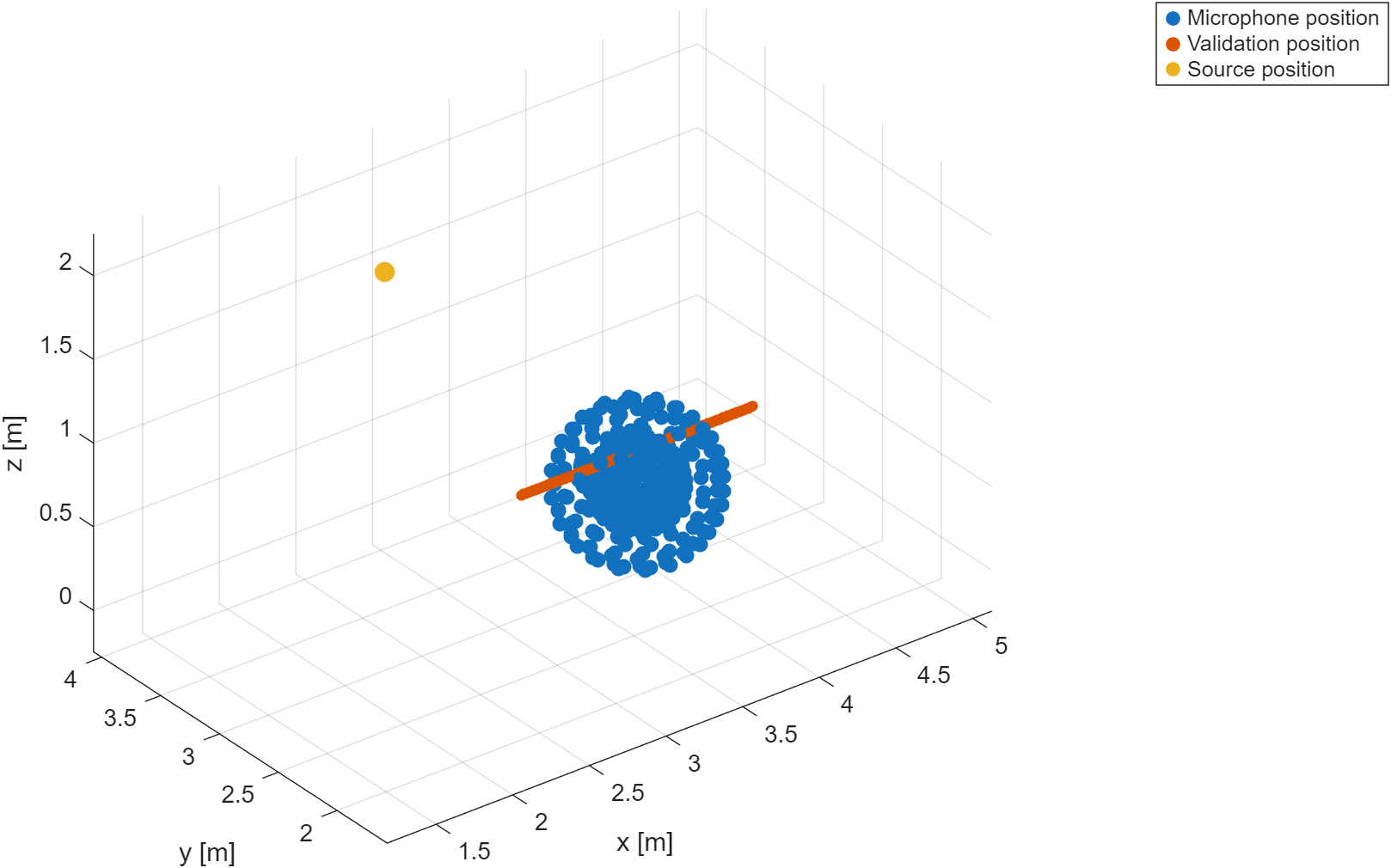}
    \caption{Geometry of the measured DTU setup. Spherical-array microphones are used as candidate observations and the linear array is used as the validation and prediction set.}
    \label{fig:measured_geometry_dtu019}
\end{figure}

In all experiments, the source spectrum used in the covariance model is band-limited and white over $70$--$1000$ Hz. Since the posterior mean depends on the ratio \(\sigma^2/q\), we fix \(q=1\) and tune \(\sigma^2\).
The covariance integral is evaluated on a Fibonacci grid with $Q=1000$ points on a sphere of radius $a=5.0$ m, unless otherwise stated. The source sphere $\SS_a$ is centered in the center of the microphone array.
The methods are evaluated using signals of $T=2000$ temporal samples. For RIR-based simulations, the first $800$ convolution samples are discarded before extracting the observation segment.

To understand the properties of the proposed approach, denoted Spatio-temporal, we compare it with several reference methods. As an offline performance reference, we first include a full-signal frequency domain kernel ridge regression estimator, denoted FD-KRR-Full, using the diffuse kernel in Proposition~\ref{prop:diffuse_freq_limit}, similarly to \cite{brunnstrom2025time,ueno2018kernel}. 
Since this estimator uses the full observation record and processes the resulting high-resolution frequency bins independently, the effect of finite-window cross-frequency covariance is expected to be negligible, 
and the method is used as an offline full-record performance reference.
We then consider three finite-window frequency domain baselines. 
The causal FD-KRR baseline applies the same bin-wise diffuse-kernel reconstruction to the length $W$ DFT of the causal window available at each time step, and the reconstructed time domain window is obtained by inverse DFT. The current sample is then taken as the output. In this case, the frequency resolution is determined by the window length and is therefore very coarse for small \(W\). 
To separate the effect of finite-windowing from the effect of causality, we also include a similar non-causal variant, denoted FD-KRR (non-causal), which uses a centered window of length \(2W-1\). 
In addition, motivated by the causal implementation used in \cite{Koyama2021}, we include FD-KRR-Trunc. This estimator is obtained by first constructing a high-resolution frequency domain kernel estimator using a grid of $T$ frequencies, transforming it to a time-domain FIR filter, and then truncating the filter so that only the causal finite-window part of the observation is retained. 
FD-KRR-Trunc is included to separate the grid resolution effect in the FD-KRR method from the effect of truncating a frequency domain estimator to a causal FIR window.
Finally, we include a purely spatial baseline, denoted Spatial, obtained from \eqref{eq:post_mean} using only the zero-lag spatial covariance, corresponding to \(W=1\). This baseline is included as a reference method to illustrate the benefit of exploiting temporal correlations in addition to the instantaneous spatial covariance.

\begin{figure}[t]
    \centering
    \includegraphics[width=\linewidth]{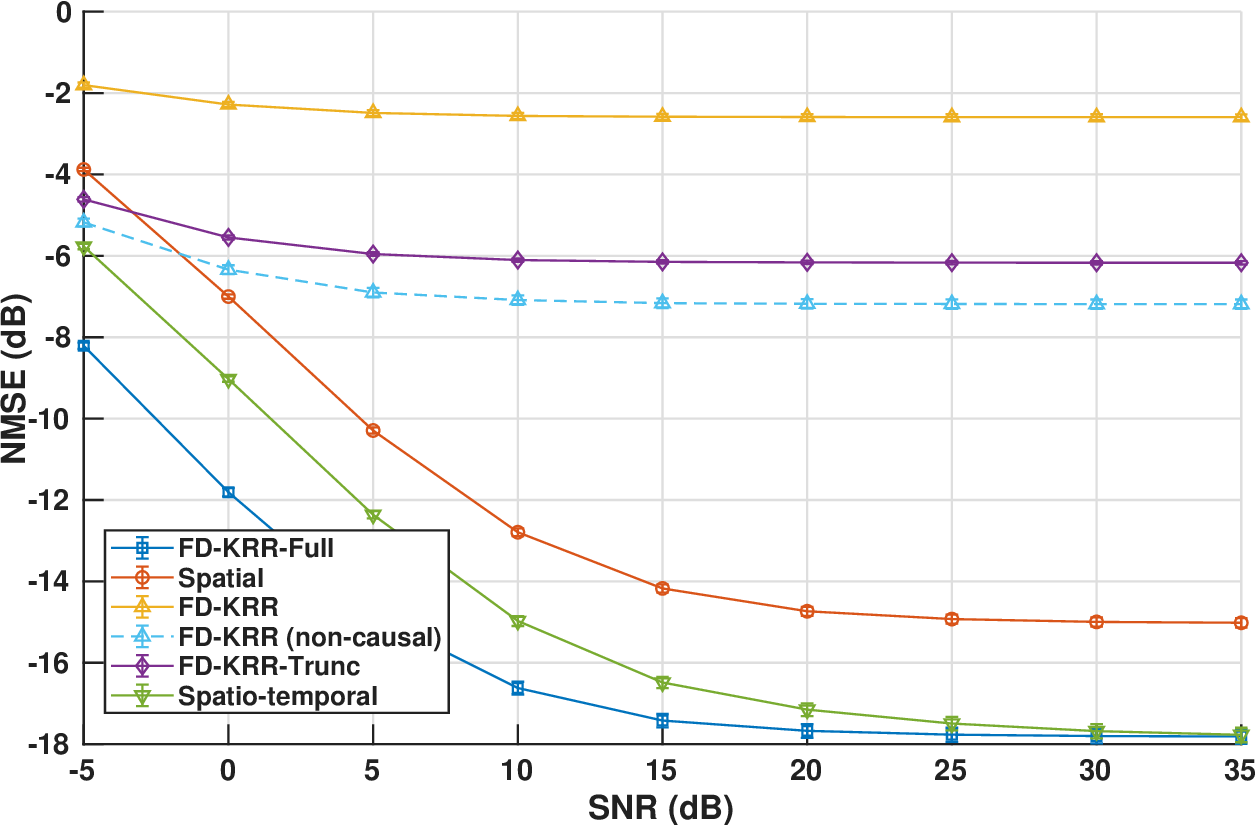}
    \caption{Reconstruction NMSE as a function of SNR for the simulated diffuse-field experiment. Causal methods use \(W=10\).}
    \label{fig:nmse_vs_snr_diffuse}
\end{figure}

The measurement noise variance $\sigma^2$ is in general unknown; here, we select $\sigma^2$ from 20 logarithmically spaced values in $[10^{-9},1]$ independently for each method and Monte Carlo realization. 
The selection is performed by leave-one-microphone-out cross-validation, minimizing the squared prediction error of the held-out microphone. The validation data are never used during this procedure - the held-out microphone is from the circular array for the simulated experiments and from the spherical array for the experiments on the measured data. For the frequency domain methods, the same criterion is summed over the selected passband frequency bins. 
The NMSE is computed as $\lVert \hat{\bu}-\bu\rVert_F^2/\lVert\bu\rVert_F^2$ after excluding the first and last $200$ samples to avoid signal start and end effects. Reported intervals are $95\%$ confidence intervals computed over $50$ Monte Carlo realizations for all numerical experiments.

\begin{figure}[t]
    \centering
    \includegraphics[width=\linewidth]{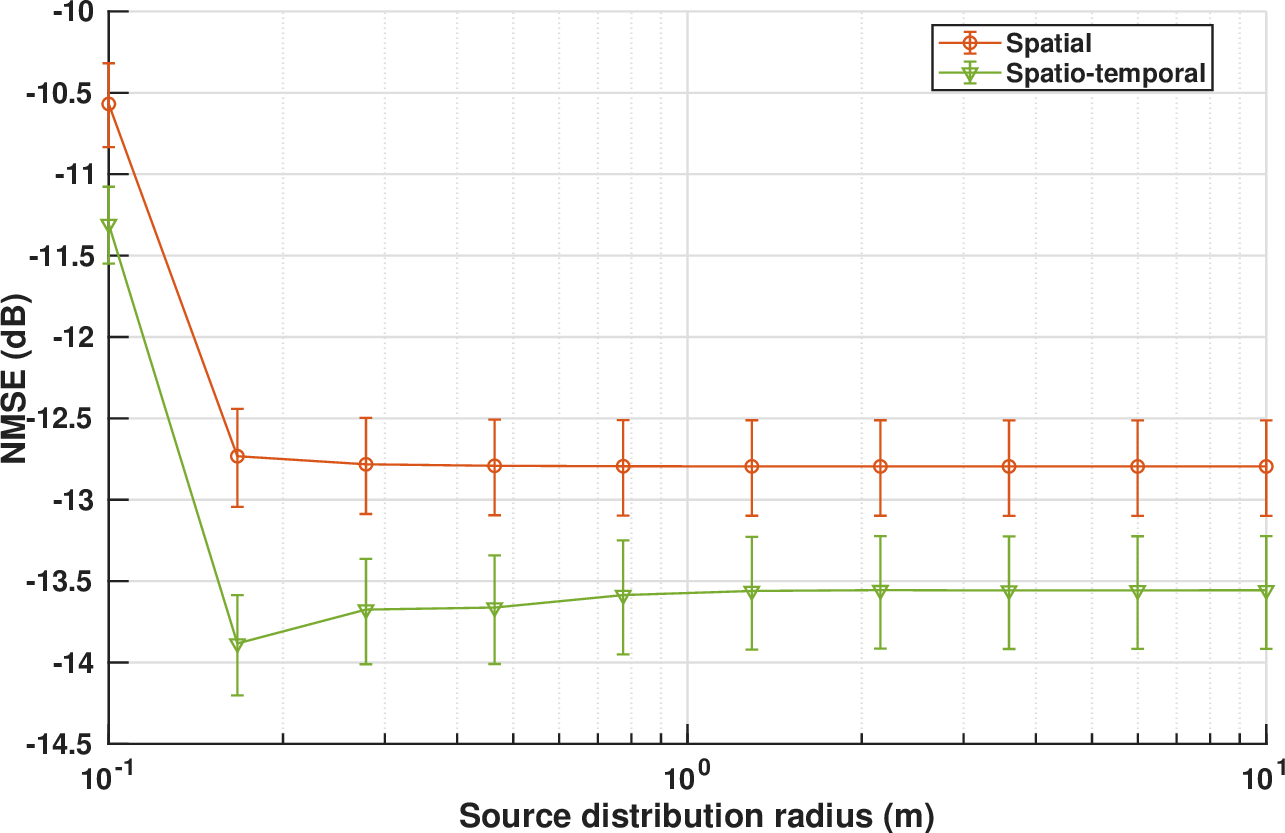}
    \caption{Sensitivity of the proposed estimator to the assumed source-sphere radius \(a\).}
    \label{fig:nmse_vs_source_radius}
\end{figure}

\subsection{Evaluation on simulated data}
The effect of finite windowing is illustrated in Figure~\ref{fig:windowed_csd_ism_pair12}, which shows the normalized magnitude
\(|\bar C_{12}^{(W)}[k,\ell]|\) of the finite-window cross-frequency covariance between microphone 1 and 2.
For short windows, the windowed observations contain substantial cross-correlation between frequency bins. 
This is the finite-window effect that motivates the spatio-temporal formulation presented herein, i.e.,
it is not reasonable to assume independence between frequencies 
when using short windows.
To make this point clear for the sound field reconstruction problem, Figure~\ref{fig:nmse_vs_winlen_diffuse} shows the corresponding NMSE of the reconstruction in the idealistic diffuse free-field simulation. The window-length sweep uses ten values between $W=2$ and $W=100$. 
The non-causal estimator use \(2W-1\) samples and are therefore not matched to the causal estimators in sample count, but are instead included to illustrate the effect of access to future observations in addition to the causal window.
The causal spatio-temporal estimator approaches the full-observation frequency domain reference already for $W=5$, whereas the finite-window frequency domain estimators require substantially longer windows to obtain comparable performance. 
The rapid convergence of the causal estimator in the diffuse simulation is expected because the data are generated from the same diffuse prior family used by the estimator and because the prediction region lies inside the relatively small microphone aperture.
This supports the interpretation that the relevant limitation is not the frequency domain representation itself, but the independence approximation induced when short causal windows are processed frequency by frequency.

The advantage of short causal windows comes with increased sensitivity to measurement noise, since the estimator uses fewer temporal observations. 
This effect is shown in Figure~\ref{fig:nmse_vs_snr_diffuse}, where the diffuse-field experiment is repeated for SNR values from $-5$ dB to $35$ dB. 
The SNR is defined as
\[
\mathrm{SNR}
=
10\log_{10}\left(\frac{\sigma_{\mathrm{sig}}^2}{\sigma_{\mathrm{add}}^2}\right),
\]
where \(\sigma_{\mathrm{sig}}^2\) is the variance of the clean microphone signal and \(\sigma_{\mathrm{add}}^2\) is the variance of the added white Gaussian noise. The noise parameter \(\sigma^2\) used in the estimator in \eqref{eq:post_mean} is here selected independently for each method, SNR condition, and Monte Carlo realization using leave-one-microphone-out cross-validation.
The causal spatio-temporal and causal FD-KRR methods use $W=10$, while the finite-window non-causal method use windows of length $19$.
At low SNR, the gap between the frequency domain reference method on the full signal and the causal spatio-temporal estimator increases, whereas the NMSE is similar for high SNR.

We also study the sensitivity to the radius of the source prior and the numerical quadrature resolution used to compute the covariance in Figures~\ref{fig:nmse_vs_source_radius} and~\ref{fig:nmse_vs_num_points_quadrature_resolution}. 
The source radius is varied over ten logarithmically spaced values between $0.1$ m and $10$ m, and the quadrature resolution is varied over ten values between $1$ and $100$ points. The results show that the covariance evaluation is numerically stable for moderate quadrature resolutions and that the reconstruction is not critically dependent on a finely tuned source radius. Although the reconstruction is essentially stable around \(Q=100\), we use \(Q=1000\) in the remaining experiments to make quadrature error negligible.

\begin{figure}[t]
    \centering
    \includegraphics[width=\linewidth]{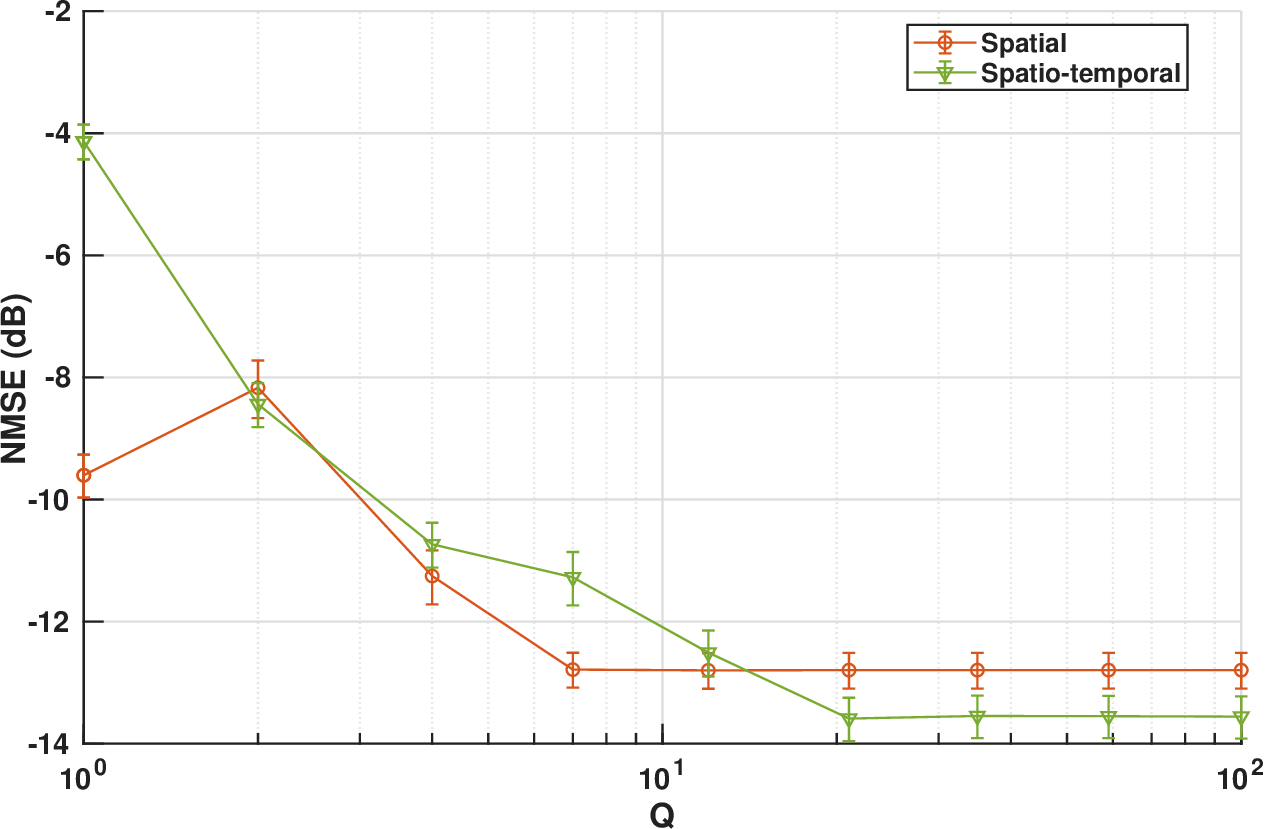}
    \caption{Sensitivity of the proposed estimator to the number of quadrature points \(Q\) used in \eqref{eq:C_quadrature}.}
    \label{fig:nmse_vs_num_points_quadrature_resolution}
\end{figure}

\subsection{Evaluation on measured data}
Next, we evaluate whether the same behavior is observed for measured room responses. 
Figure~\ref{fig:nmse_vs_winlen_measured_dtu} shows the NMSE as a function of window length for the measured dataset, using ten logarithmically spaced window settings between $W=2$ and $W=20$. 
The causal spatio-temporal formulation gives substantially lower reconstruction error than the finite-window frequency domain estimators. 
Although the error decays as a function of $W$, there is still a significant gap to the FD-KRR-Full reference at $W=20$. Since increasing $W$ further would shrink this gap at the cost of increasing the computational complexity, this further motivates the need for careful selection of the spatio-temporal observations.

Furthermore, we study how the reconstruction varies along the linear validation array in Figure~\ref{fig:validation_distance_measured_dtu}. 
The dashed markers indicate the transition from interpolation inside the spherical measurement aperture to extrapolation outside it. The causal spatio-temporal estimator follows the NMSE of the full-observation frequency domain reference closely both inside and outside the measured aperture, although the extrapolation region is more challenging. 
Since the posterior variance is used for the optimal spatio-temporal selection, it is relevant to understand how well calibrated the posterior variance of the model is. 
We illustrate this in Figure~\ref{fig:validation_distance_measured_dtu_causal_variance}, showing that the posterior variance is of the same order as the empirical error in the interpolation region, while becoming conservative in the extrapolation region.
This behavior is consistent with the transition from interpolation, where the prediction points are surrounded by observations, to extrapolation, where the reconstruction relies more strongly on the prior.
The plotted posterior variance is the diagonal of \eqref{eq:post_var} for the causal spatio-temporal estimator with $W=20$, averaged over the Monte Carlo microphone splits. It is compared with the empirical error variance
\begin{align}
    \hat v^{\mathrm{emp}}_p
    =
    \frac{1}{N_t}
    \sum_{n\in\mathcal{T}}
    e_p[n]^2,
    \label{eq:empirical_error_variance}
\end{align}
where \(e_p[n]=\hat u_p[n]-u_p[n]\), $\mathcal{T}$ is the set of time indices used for validation, and $N_t=|\mathcal{T}|$.
Overall, the predictive variance is on the same order of magnitude as the empirical, suggesting that this objective is meaningful for the spatio-temporal sampling problem further evaluated in the next section.

\begin{figure}[t]
    \centering
    \includegraphics[width=\linewidth]{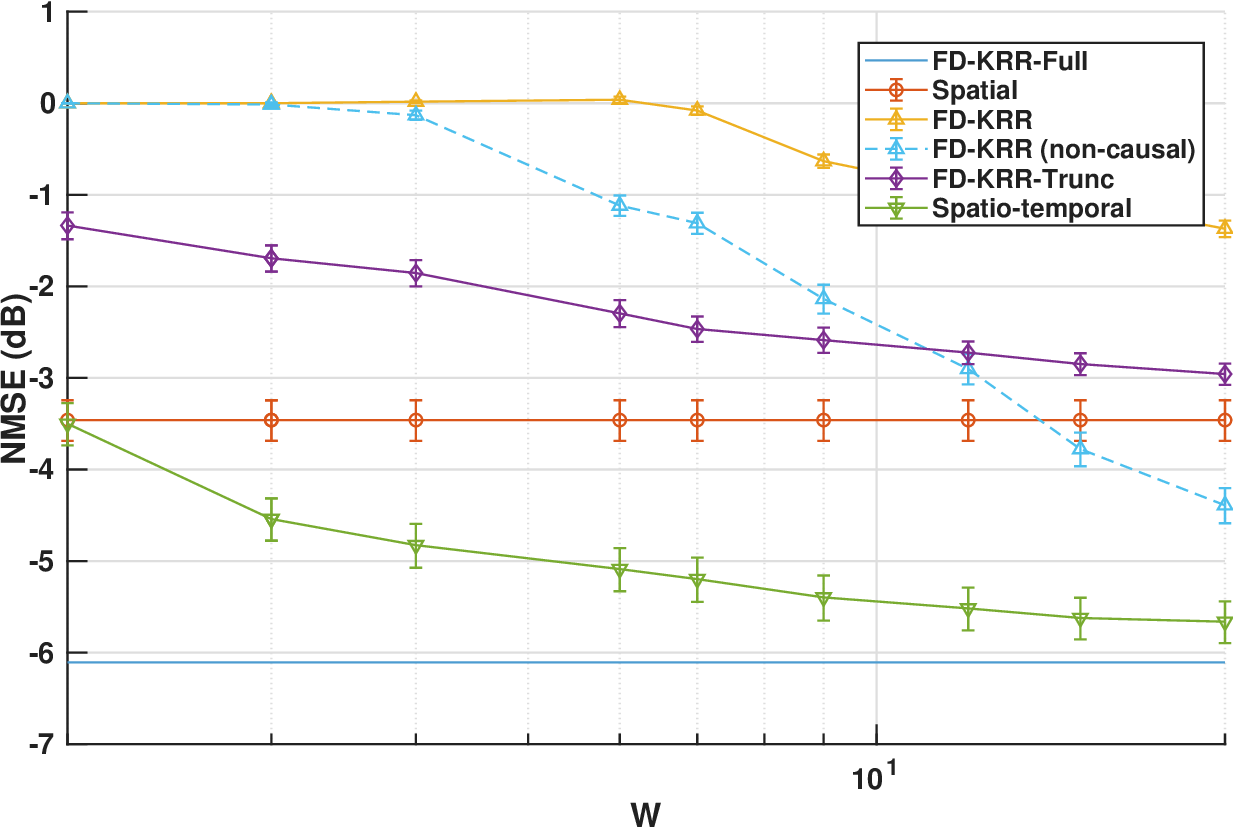}
    \caption{Illustration of the NMSE as a function of the window length $W$ for the measured dataset. The non-causal finite-window FD-KRR methods use windows of length $2W-1$, while the causal methods use length $W$.}
    \label{fig:nmse_vs_winlen_measured_dtu}
\end{figure}

\begin{figure*}[t]
    \centering
    \begin{subfigure}[t]{0.49\textwidth}
        \centering
        \includegraphics[width=\linewidth]{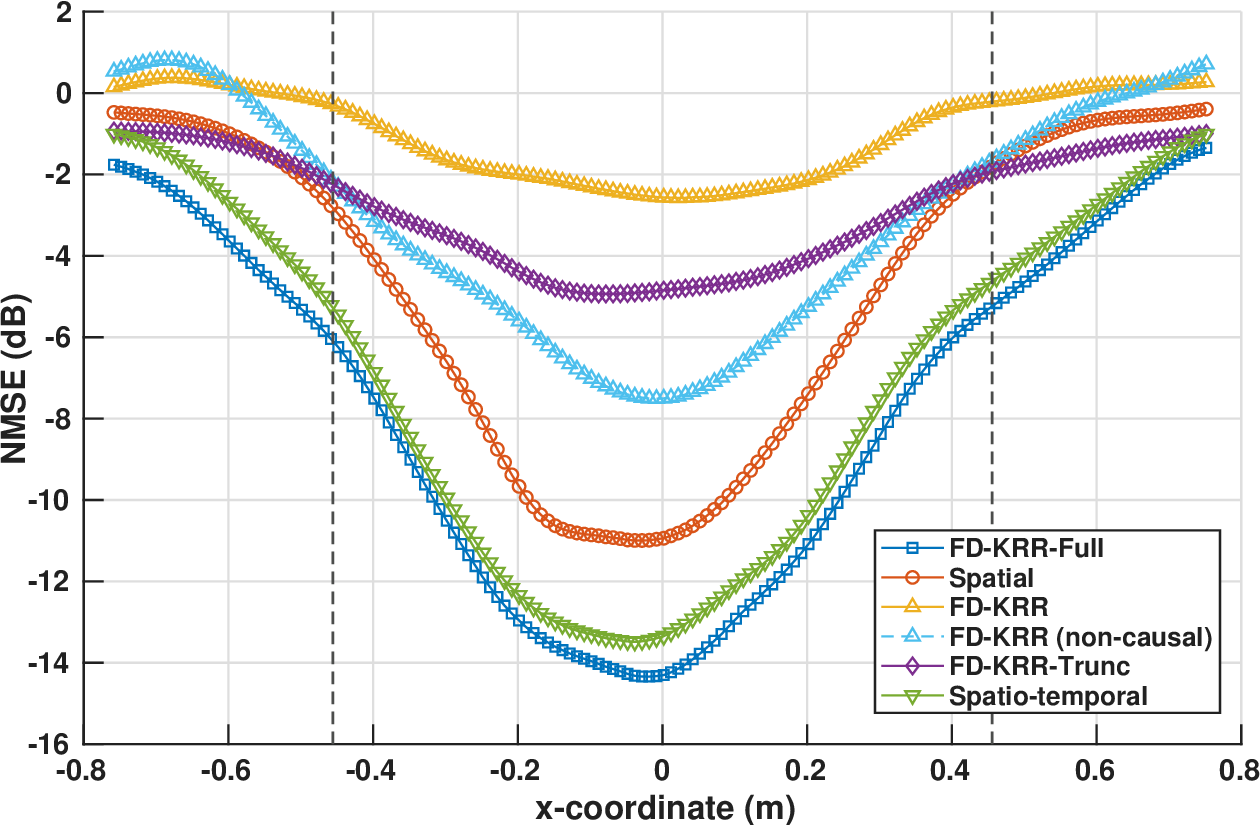}
        \caption{}
        \label{fig:validation_distance_measured_dtu_nmse}
    \end{subfigure}
    \hfill
    \begin{subfigure}[t]{0.49\textwidth}
        \centering
        \includegraphics[width=\linewidth]{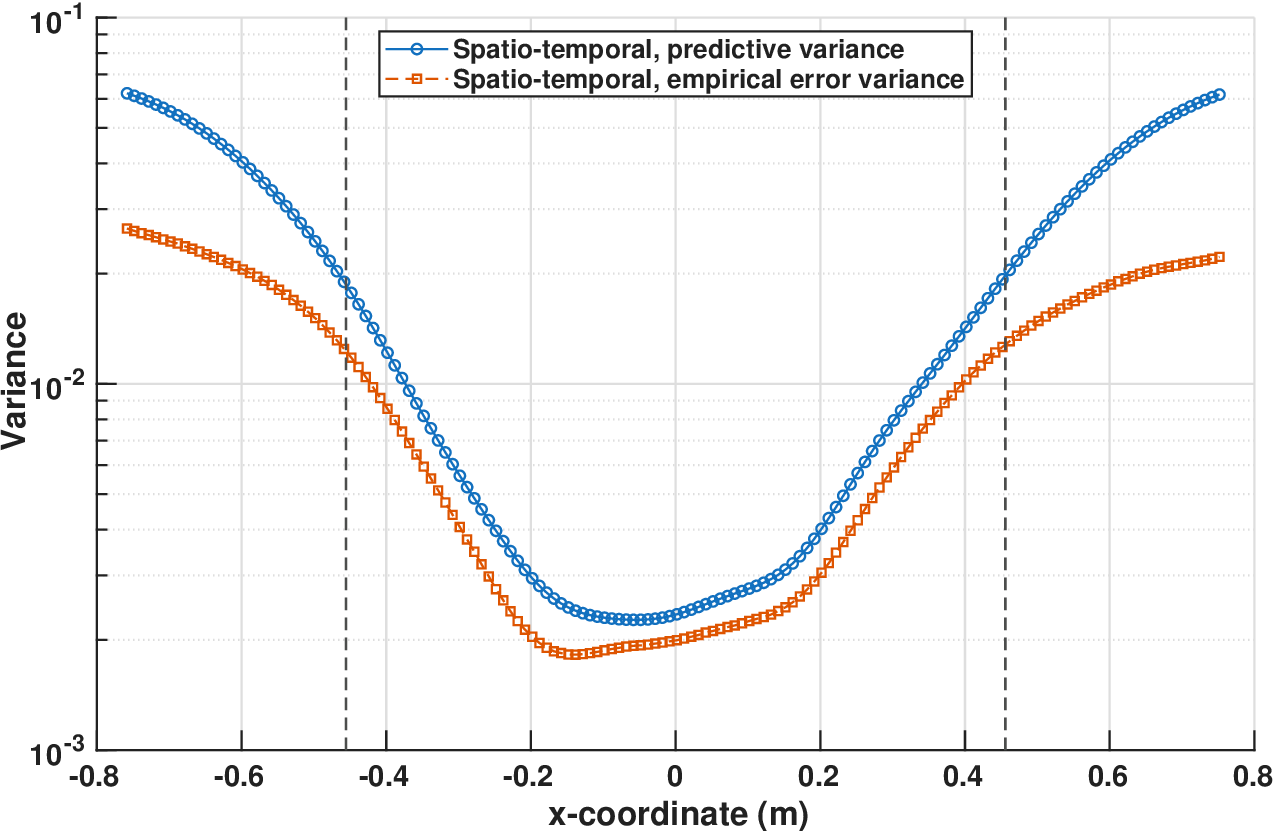}
        \caption{}
        \label{fig:validation_distance_measured_dtu_causal_variance}
    \end{subfigure}
    \caption{Reconstruction along the linear validation array illustrated in Figure~\ref{fig:measured_geometry_dtu019}. Dashed vertical lines mark the spherical-array aperture boundary separating interpolation and extrapolation. The NMSE is illustrated in a), and the posterior variance of the proposed spatio-temporal causal estimator is compared to the empirical prediction variance in b). }
    \label{fig:validation_distance_measured_dtu}
\end{figure*}

\subsection{Spatio-temporal sample selection}
Finally, we evaluate the optimal spatio-temporal sampling formulation on the measured DTU data. The candidate horizon is $W=20$, so the full causal estimator uses $K=MW=1000$ possible spatio-temporal samples.  
We set $\varepsilon$ to $10^{-9}$, $\rho$ to $1.2$, and run the projected gradient solver detailed in Appendix \ref{app:projected_gradient} for 100 iterations.
The prediction points in the posterior-variance objective \eqref{eq:sampling_problem} are all positions on the measured linear validation array. Note that only the coordinates of the linear array positions are used in the objective and not the measured pressure values at those positions.
We compare the method described in Section~\ref{sec:optimal_sampling} with two baselines, recent and random selection.
The recent baseline selects all microphones at the most recent time lags until fewer than \(M\) samples remain, and then selects the remaining samples chronologically from the next lag. This corresponds to a typical naive implementation in a real-time system. The second baseline selects samples uniformly at random without replacement for each Monte Carlo realization.
The qualitative behavior of the selected samples is shown in Figure~\ref{fig:nmse_vs_k_measured_selection_patterns} for $K=50$, $K=189$, and $K=1000$. 
Note that the selected patterns also indicate which microphones contribute most strongly to the reconstruction, suggesting a possible extension to joint array-geometry and temporal-sampling design, similar to the purpose in \cite{verburg2024optimal}.
Note also that there seems to be a periodic structure of the number of selected microphones as a function of lag for the proposed method in Figure \ref{fig:nmse_vs_k_measured_selection_patterns}. We confirm this behavior in Figure \ref{fig:num_samples_per_lag}, which illustrates the average number of selected microphones as a function of lag based on $50$ Monte Carlo realizations of the experiment. For the $12$ first lags, there is a periodic structure with a period of about $4$, corresponding to a frequency of $2$kHz, i.e., the lowest sampling frequency necessary to represent the assumed band-width $70-1000$ Hz of the source distribution in \eqref{eq:Phi_bp}. This indicates a strong influence of the spectral model on the optimal sampling scheme, consistent with recent results in \cite{juhlin2023optimal}.
Figure~\ref{fig:nmse_vs_k_measured} shows the resulting NMSE as a function of the sampling budget, with $K$ varied over ten logarithmically spaced budgets from $50$ to $1000$. 
For a fixed reconstruction error, the proposed selection requires substantially fewer samples than the random and recent baselines. 
For example, for an NMSE of approximately -4 to -5 dB, the same accuracy can be obtained when only half of the observations are used.
If \(K=\alpha MW\), the factorization, memory, and precomputed online filtering costs scale as \(\alpha^3\), \(\alpha^2\), and \(\alpha\), respectively, relative to the full estimator. Thus, using half of the spatio-temporal observations reduces the factorization cost by approximately a factor of eight, the memory cost by a factor of four, and the online filtering cost by a factor of two.

\begin{figure*}[t]
    \centering
    \begin{subfigure}[t]{0.32\textwidth}
        \centering
        \includegraphics[width=\linewidth]{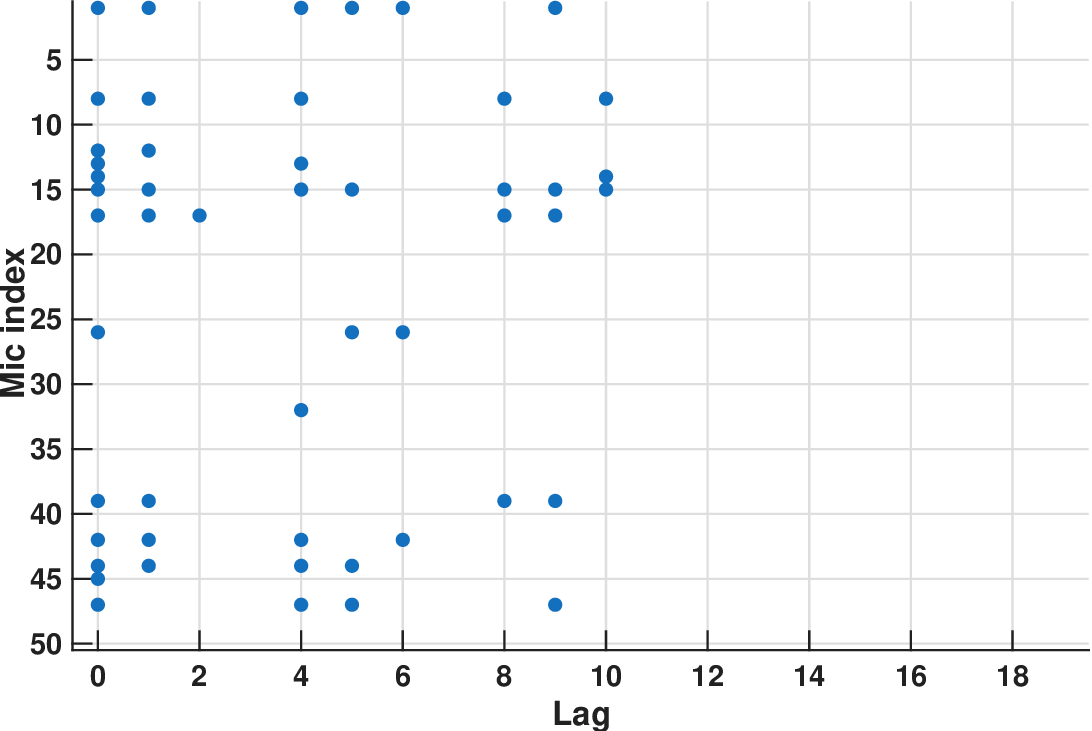}
        \caption{Proposed, $K=50$.}
        \label{fig:nmse_vs_k_measured_selection_patterns_proposed_k50}
    \end{subfigure}
    \hfill
    \begin{subfigure}[t]{0.32\textwidth}
        \centering
        \includegraphics[width=\linewidth]{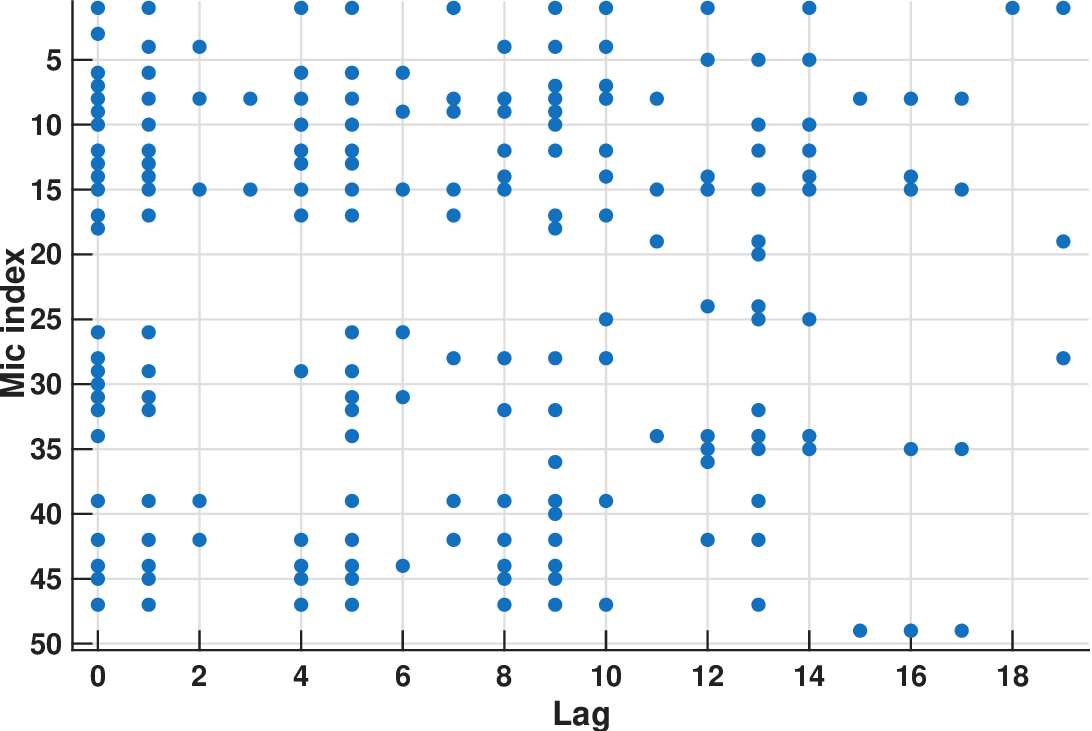}
        \caption{Proposed, $K=189$.}
        \label{fig:nmse_vs_k_measured_selection_patterns_proposed_k189}
    \end{subfigure}
    \hfill
    \begin{subfigure}[t]{0.32\textwidth}
        \centering
        \includegraphics[width=\linewidth]{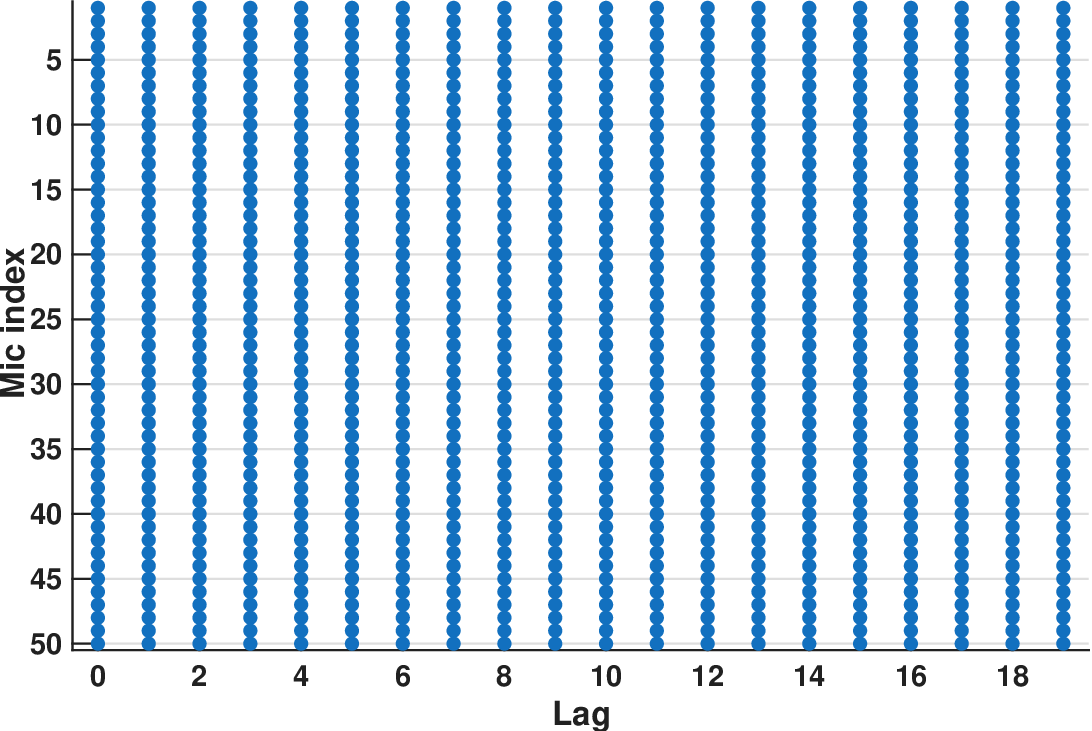}
        \caption{Proposed, $K=1000$.}
        \label{fig:nmse_vs_k_measured_selection_patterns_proposed_k1000}
    \end{subfigure}

    \vspace{0.6em}

    \begin{subfigure}[t]{0.32\textwidth}
        \centering
        \includegraphics[width=\linewidth]{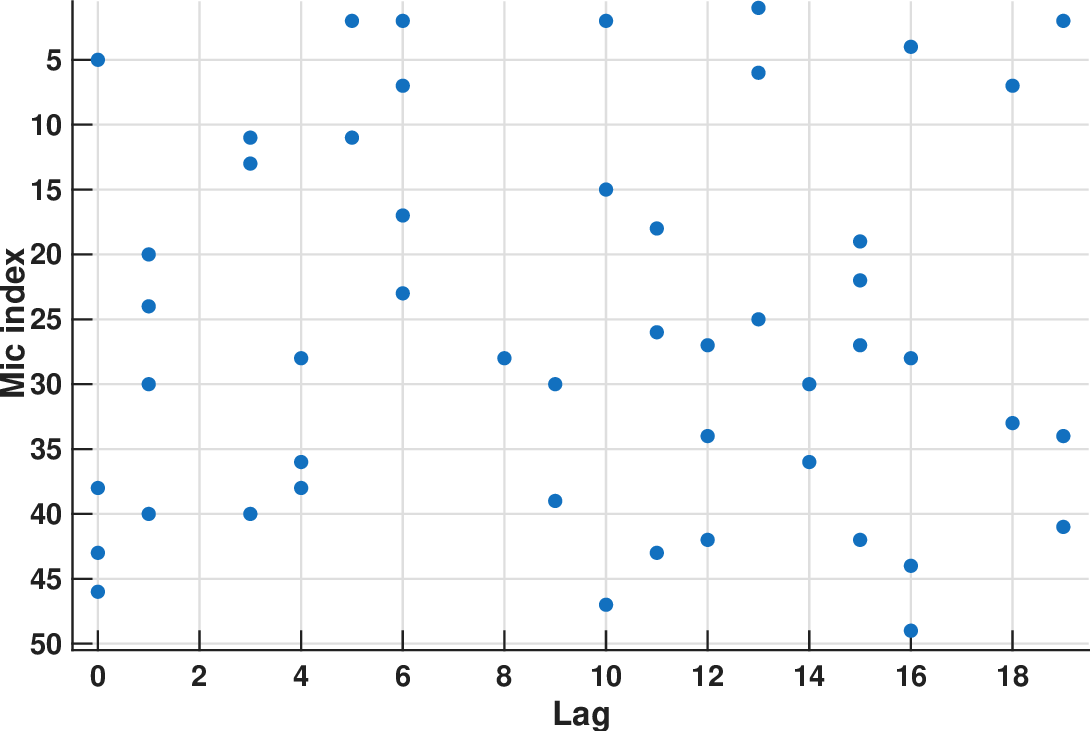}
        \caption{Random, $K=50$.}
        \label{fig:nmse_vs_k_measured_selection_patterns_random_k50}
    \end{subfigure}
    \hfill
    \begin{subfigure}[t]{0.32\textwidth}
        \centering
        \includegraphics[width=\linewidth]{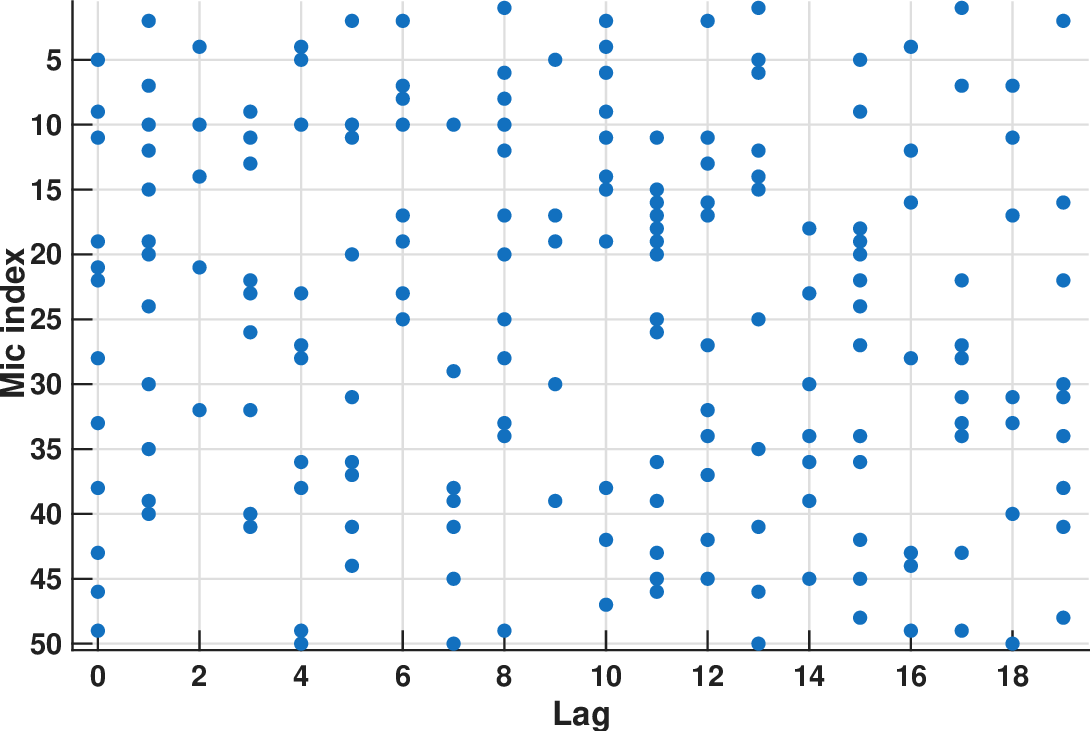}
        \caption{Random, $K=189$.}
        \label{fig:nmse_vs_k_measured_selection_patterns_random_k189}
    \end{subfigure}
    \hfill
    \begin{subfigure}[t]{0.32\textwidth}
        \centering
        \includegraphics[width=\linewidth]{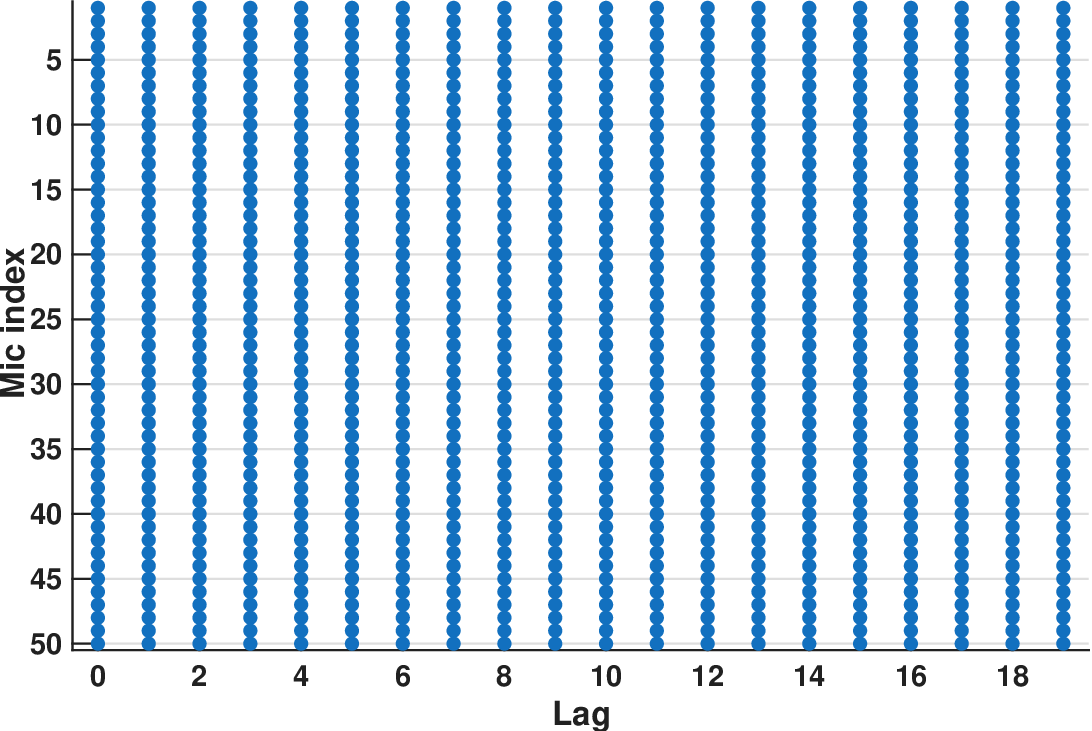}
        \caption{Random, $K=1000$.}
        \label{fig:nmse_vs_k_measured_selection_patterns_random_k1000}
    \end{subfigure}

    \vspace{0.6em}

    \begin{subfigure}[t]{0.32\textwidth}
        \centering
        \includegraphics[width=\linewidth]{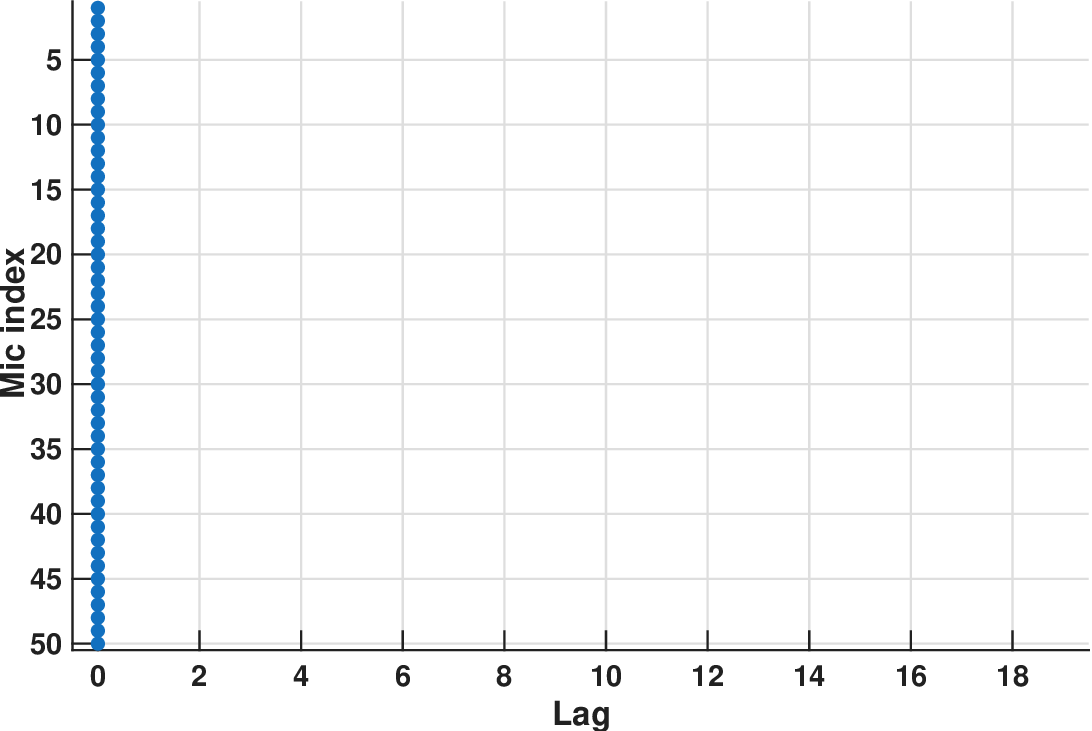}
        \caption{Recent, $K=50$.}
        \label{fig:nmse_vs_k_measured_selection_patterns_recent_k50}
    \end{subfigure}
    \hfill
    \begin{subfigure}[t]{0.32\textwidth}
        \centering
        \includegraphics[width=\linewidth]{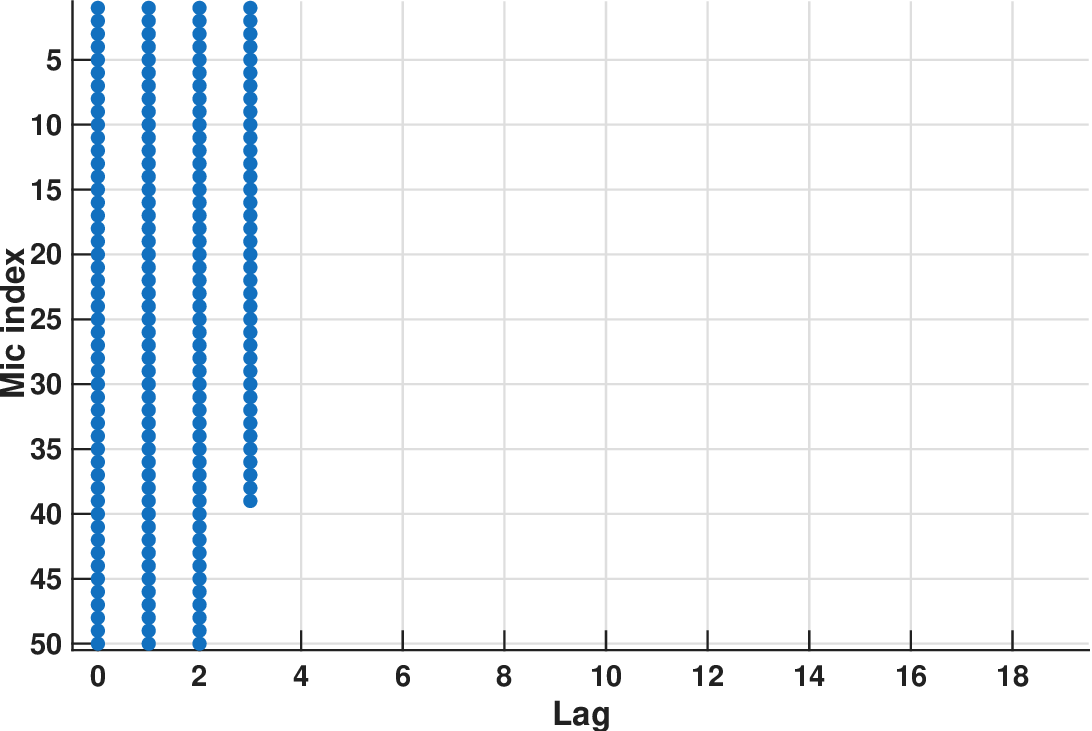}
        \caption{Recent, $K=189$.}
        \label{fig:nmse_vs_k_measured_selection_patterns_recent_k189}
    \end{subfigure}
    \hfill
    \begin{subfigure}[t]{0.32\textwidth}
        \centering
        \includegraphics[width=\linewidth]{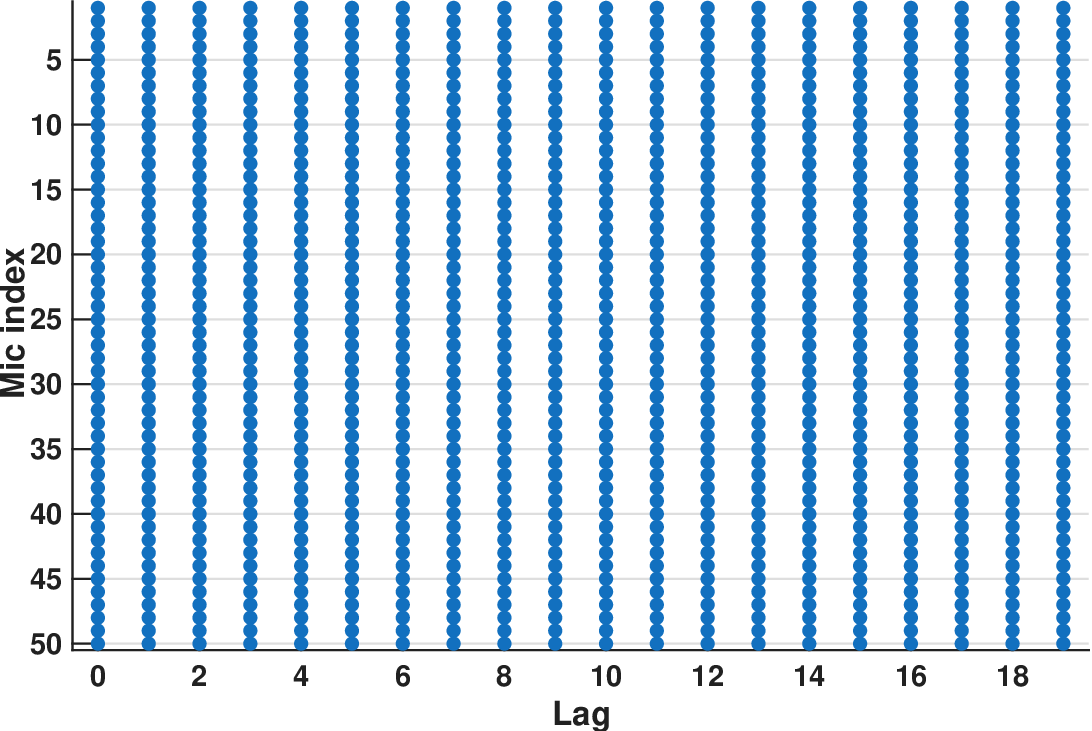}
        \caption{Recent, $K=1000$.}
        \label{fig:nmse_vs_k_measured_selection_patterns_recent_k1000}
    \end{subfigure}
    \caption{Illustration of the spatio-temporal sample selection patterns for the proposed, random, and recent strategies. Candidate observations are spherical-array microphone samples over the causal horizon $W=20$.}
    \label{fig:nmse_vs_k_measured_selection_patterns}
\end{figure*}

\begin{figure}[t]
    \centering
    \includegraphics[width=\linewidth]{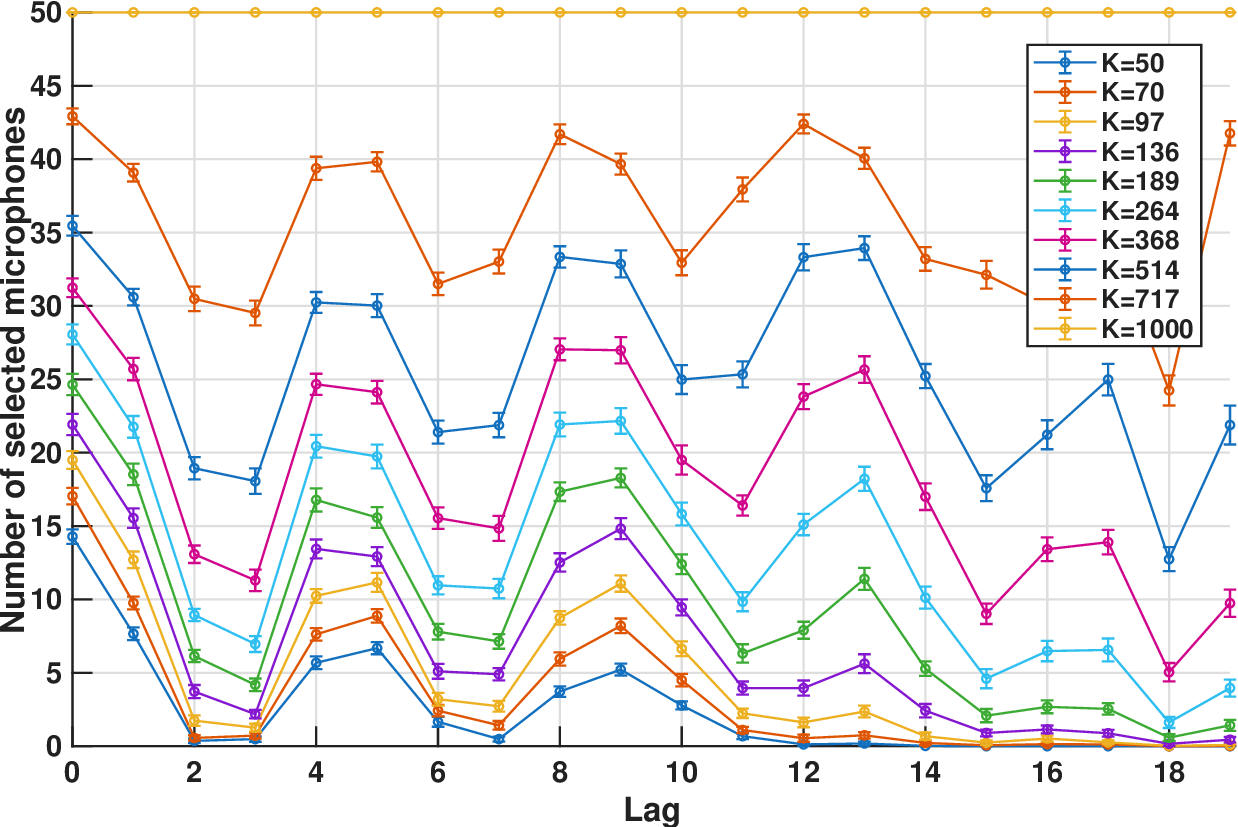}
    \caption{Illustration of the number of selected microphones per lag when using the proposed spatio-temporal selection method for 10 different values of $K$. There is a consistent oscillatory pattern that is related to the assumed spectral structure of the source model in \eqref{eq:Phi_bp}.}
    \label{fig:num_samples_per_lag}
\end{figure}

\begin{figure}[t]
    \centering
    \includegraphics[width=\linewidth]{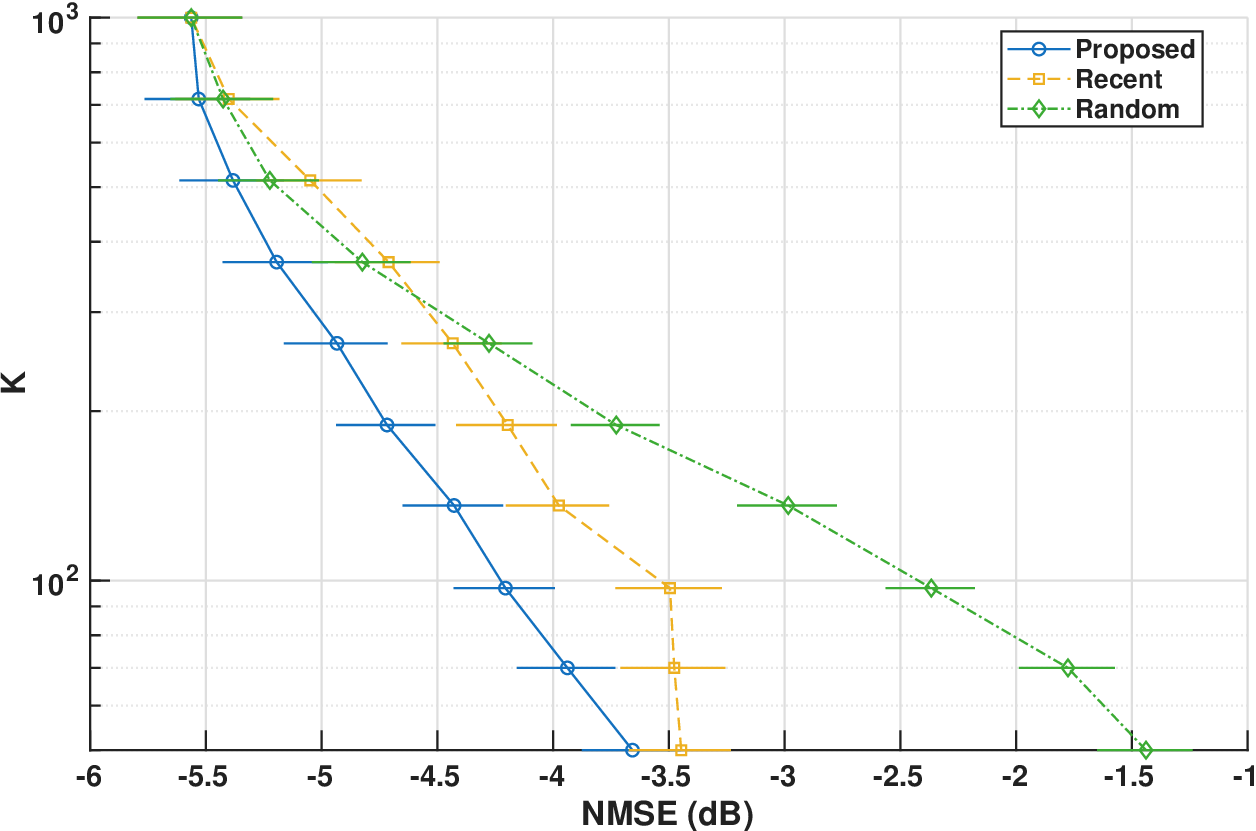}
    \caption{Illustration of the number of observations $K$ needed to obtain a specific reconstruction NMSE evaluated on the DTU dataset. Error bars show $95\%$ confidence intervals over Monte Carlo microphone splits. Substantially fewer samples are needed for the same reconstruction accuracy when using the proposed sampling scheme.}
    \label{fig:nmse_vs_k_measured}
\end{figure}

%=========================================================

\section{Conclusion}
\label{sec:conclusion}
We formulated causal finite-window sound field reconstruction as a spatio-temporal LMMSE estimation problem with a covariance induced by the wave equation driven by a temporally stationary stochastic source distribution. For a spatially white source on a surrounding sphere, the covariance is numerically tractable and converges in the far-field limit to the classical diffuse-field coherence model. 
The formulation retains the temporal correlations present in short causal windows and explains why frequency domain estimators that process finite-window frequency bins independently can lose reconstruction accuracy.
Experiments on simulated and measured sound fields showed improved short window reconstruction compared with finite-window frequency domain baselines. Finally, the proposed sample selection scheme based on the posterior variance achieved similar reconstruction accuracy using substantially fewer spatio-temporal observations, reducing factorization, memory, and online filtering costs. These results support the use of causal spatio-temporal covariance models for low-latency sound-field reconstruction in real-time control and rendering applications.

%=========================================================
\appendices
\section{Derivation of the stationary spatio-temporal covariance}
\label{app:u_covariance_derivation}
The spatio-temporal covariance of the sound field in \eqref{eq:u_covariance} is shown to be given on the form \eqref{eq:Cu_integral}. We first emphasize the importance of the stationarity in time of the source covariance in \eqref{eq:source_cov}. Using \eqref{eq:source_cov} in \eqref{eq:u_covariance} yields
\begin{align}
  &C_u\!\big((t,\br),(t',\br')\big)
  = \frac{q}{16\pi^2}
  \int_{\RR^3}\int_{\RR^3}
  \frac{\wt(\br_0,\br_1)}
  {\|\br-\br_0\|_2\,\|\br'-\br_1\|_2} \nonumber \\
  &\times\int_{\RR}\int_{\RR}
  \delta\!\Big(t-\tau-\frac{\|\br-\br_0\|_2}{c}\Big)
  \delta\!\Big(t'-\tau'-\frac{\|\br'-\br_1\|_2}{c}\Big) \nonumber\\ 
  &\quad \qquad \qquad \qquad  \qquad \times \kappa(\tau-\tau') d\tau\,d\tau'\,d\br_0\,d\br_1,
\end{align}
where the integrals involving the temporal dimensions can be further simplified as a result of the assumed temporal stationarity to
\begin{align}
  C_u&\!\big((t,\br),(t',\br')\big)
  = \frac{q}{16\pi^2}
  \int_{\RR^3}\int_{\RR^3}
  \frac{\wt(\br_0,\br_1)}
  {\|\br-\br_0\|_2\,\|\br'-\br_1\|_2}
  \nonumber\\
  &\times\kappa\!\Bigg(
    (t-t')
    -
    \frac{\|\br-\br_0\|_2-\|\br'-\br_1\|_2}{c}
  \Bigg)
  \,d\br_0\,d\br_1.
  \label{eq:Cu_after_time_stationarity}
\end{align}
Note that this form of the covariance function could already be numerically tractable if the source distribution were restricted to a bounded spatial region. In the following, we consider the special case of a source distribution supported on the spherical surface \(\SS_a\).
Furthermore, we specialize to the spatially white source model in \eqref{eq:source_cov_2}, i.e.,
\begin{align}
  \wt(\br_0,\br_1)=\delta(\br_0-\br_1),
  \qquad \br_0,\br_1\in\SS_a,
\end{align}
such that \eqref{eq:Cu_after_time_stationarity} becomes
\begin{align}
  C_u&\!\big((t,\br),(t',\br')\big)
  =
  \frac{q}{16\pi^2}
  \int_{\SS_a}
  \frac{1}
  {\|\br-\br_0\|_2\,\|\br'-\br_0\|_2} \nonumber\\
  &\times\kappa\!\Bigg(
    (t-t')
    -
    \frac{\|\br-\br_0\|_2-\|\br'-\br_0\|_2}{c}
  \Bigg)\,dS(\br_0),
\end{align}
which is \eqref{eq:Cu_integral}.

\section{Connection to frequency domain diffuse models}
\label{app:freq_domain_connection}

Taking the Fourier transform of \eqref{eq:Cu_integral} with respect to $\tau$ yields
\begin{align}
  &S_{u,a}(\br,\br';\omega)
  =
  \int_{\RR} R_u(\br,\br';\tau)e^{-j\omega\tau}\,d\tau
  \nonumber\\
  &=
  \frac{q}{16\pi^2}
  \int_{\SS_a}
  \frac{
    \int_{\RR}
    \kappa\!\left(
      \tau - \frac{\|\br-\br_0\|_2-\|\br'-\br_0\|_2}{c}
    \right)e^{-j\omega\tau}\,d\tau
  }{
    \|\br-\br_0\|_2\,\|\br'-\br_0\|_2
  }
  \,dS(\br_0).
\end{align}
Using the Wiener-Khinchin theorem (see for example \cite{hayes1996statistical})
\begin{align}
  \int_{\RR}\kappa(\tau-\delta)e^{-j\omega\tau}\,d\tau
  =
  \Phi(\omega)e^{-j\omega\delta},
\end{align}
and the model of the PSD in \eqref{eq:Phi_bp}, it follows that
\begin{align}
  S_{u,a}(\br,\br';\omega)
  =
  \frac{q\,\Phi(\omega)}{16\pi^2}
  \int_{\SS_a}
  \frac{
    e^{-j\frac{\omega}{c}\left(\|\br-\br_0\|_2-\|\br'-\br_0\|_2\right)}
  }{
    \|\br-\br_0\|_2\,\|\br'-\br_0\|_2
  }
  \,dS(\br_0).
  \label{eq:Su_surface_appendix}
\end{align}
Let $\br_0=a\hat{\bs}$ with $\hat{\bs}\in \mathbb S^2$, so that
$dS(\br_0) = a^2 d\Omega(\hat{\bs})$, where \(d\Omega(\hat{\bs})\) denotes the solid-angle measure on \(\mathbb S^2\). For a fixed interior point $\br,\br'$ and $a\to\infty$,
\begin{align}
  \|\br-a\hat{\bs}\|_2 = a-\hat{\bs}^\top\br + \mathcal{O}(a^{-1}),
\end{align}
and therefore
\begin{align}
  \|\br-a\hat{\bs}\|_2-\|\br'-a\hat{\bs}\|_2
  &=
  -\hat{\bs}^\top(\br-\br') + \mathcal{O}(a^{-1}),
  \\
  \frac{a^2}{
    \|\br-a\hat{\bs}\|_2\,\|\br'-a\hat{\bs}\|_2
  }
  &\to 1.
\end{align}
Thus,
\begin{align}
  S_{u,a}(\br,\br';\omega)
  \xrightarrow[a\to\infty]{}
  \frac{q\,\Phi(\omega)}{16\pi^2}
  \int_{\mathbb S^2}
  e^{j\frac{\omega}{c}\hat{\bs}^\top(\br-\br')}
  \,d\Omega(\hat{\bs}).
\end{align}
Let $d=\|\br-\br'\|_2$. Using eq. (10.1.14) in \cite{Abramowitz1972}, stating that
\begin{align}
  \int_{\mathbb S^2}
  e^{j\frac{\omega}{c}\hat{\bs}^\top(\br-\br')}
  \,d\Omega(\hat{\bs})
  =
  4\pi\,
  \frac{\sin\!\left(\frac{\omega}{c}\|\br-\br'\|_2\right)}
       {\frac{\omega}{c}\|\br-\br'\|_2},
\end{align}
 yields
\begin{align}
  S_{u,a}(\br,\br';\omega)
  \xrightarrow[a\to\infty]{}
  \frac{q\,\Phi(\omega)}{4\pi}\,
  \frac{\sin\!\left(\frac{\omega}{c}d\right)}
       {\frac{\omega}{c}d},
\end{align}
which proves Proposition~\ref{prop:diffuse_freq_limit}.
Evaluating the same limit at $\br=\br'$ yields
\begin{align}
  S_{u,a}(\br,\br;\omega)
  \xrightarrow[a\to\infty]{}
  \frac{q\,\Phi(\omega)}{4\pi},
\end{align}
and therefore the normalized coherence %is
\begin{align}
  \gamma(\br,\br';\omega)
  = 
  \frac{S_{u,a}(\br,\br';\omega)}
  {\sqrt{S_{u,a}(\br,\br;\omega)\,S_{u,a}(\br',\br';\omega)}}
  =
  \frac{\sin\!\left(\frac{\omega}{c}\|\br-\br'\|_2\right)}
       {\frac{\omega}{c}\|\br-\br'\|_2}.
\end{align}

\section{Finite-window frequency domain covariance}
\label{app:windowed_csd_appendix}

In this appendix, we show how a finite causal observation window induces coupling between
frequency bins even when the underlying process is stationary. Let
\(u[n]=u(t_n,\br)\) and \(u'[n]=u(t_n,\br')\) be mean-zero jointly stationary discrete-time
processes with cross-covariance
\begin{align}
  C_{uu'}[\ell]
  =
  \mathbb E[u[n]u'[n-\ell]^*],
\end{align}
and cross-spectral density
\begin{align}
  S_{uu'}(\vartheta)
  =
  \sum_{\ell\in\mathbb Z}C_{uu'}[\ell]e^{-j\vartheta\ell},
  \qquad \vartheta\in[-\pi,\pi].
\end{align}
For a rectangular causal window of length \(W\), define the DFT coefficients in chronological order as
\begin{align}
  U_k[n]
  &=
  \sum_{r=0}^{W-1}
  u[n-W+1+r]e^{-j\vartheta_k r},
  \\
  U'_\ell[n]
  &=
  \sum_{s=0}^{W-1}
  u'[n-W+1+s]e^{-j\vartheta_\ell s},
\end{align}
where \(\vartheta_k=2\pi k/W\). Then,
\begin{align}
  \mathbb E[U_k[n]U'_\ell[n]^*]
  =
  \sum_{r=0}^{W-1}\sum_{s=0}^{W-1}
  C_{uu'}[r-s]e^{-j\vartheta_k r}e^{j\vartheta_\ell s}.
\end{align}
Using
\begin{align}
  C_{uu'}[r-s]
  =
  \frac{1}{2\pi}\int_{-\pi}^{\pi}
  S_{uu'}(\vartheta)e^{j\vartheta(r-s)}\,d\vartheta,
\end{align}
and interchanging summation and integration yields
\begin{align}
  \mathbb E[U_k[n]U'_\ell[n]^*]
  &= \nonumber \\
  & \hspace{-15mm}
  \frac{1}{2\pi}\int_{-\pi}^{\pi}
  S_{uu'}(\vartheta)
  D_W(\vartheta_k-\vartheta)
  D_W^*(\vartheta_\ell-\vartheta)
  \,d\vartheta,
\end{align}
where
\begin{align}
  D_W(\vartheta)
  =
  \sum_{r=0}^{W-1}e^{-j\vartheta r}
\end{align}
is the Dirichlet kernel. Therefore, the finite-window frequency domain covariance is
\begin{align}
  C_{uu'}^{(W)}[k,\ell]
  &=
  \mathbb E[U_k[n]U'_\ell[n]^*] \nonumber \\
  & \hspace{-8mm} =
  \frac{1}{2\pi}\int_{-\pi}^{\pi}
  S_{uu'}(\vartheta)
  D_W(\vartheta_k-\vartheta)
  D_W^*(\vartheta_\ell-\vartheta)
  \,d\vartheta.
\end{align}
For finite \(W\), the off-diagonal terms with \(k\neq \ell\) are generally nonzero. Hence, the finite-window frequency domain covariance does not separate into independent frequency bins. A frequency domain estimator that retained the full cross-frequency covariance would be equivalent to the corresponding time-domain estimator up to a unitary transformation. The loss considered herein therefore comes from the approximation from assuming independent frequency bins.

\section{Derivation of projected gradient updates }
\label{app:projected_gradient}

In this appendix, we derive the gradient of the relaxed objective in
\eqref{eq:sampling_relaxed_objective} and summarize the projected-gradient
update used to solve \eqref{eq:sampling_relaxed_problem}.
Define
\begin{align}
  \mathbf{M}(\mathbf{z})
  &=
  \bK_{yy} + \sigma^2 \mathbf{Z}^{-2}, \\ 
  \bA&=
  \bK_{yu}\bK_{uy},
\end{align}
so that the relaxed objective can be written as
\begin{align}
  \phi(\mathbf{z})
  =
  \tr(\bK_{uu})
  -
  \tr\!\Big(
    \mathbf{M}(\mathbf{z})^{-1}\bA
  \Big).
\end{align}
Since the first term is constant, only the second term contributes to the
gradient. 
Furthermore, for each component $z_i$, define
\begin{align}
  d_i(\mathbf{z}) = \sigma^2 z_i^{-2},
  \qquad
  d_i'(\mathbf{z}) = -2\sigma^2 z_i^{-3}.
\end{align}
As $\mathbf{M}(\mathbf{z})$ depends on $z_i$ only through its $i$th diagonal
entry, %we have
\begin{align}
  \frac{\partial \mathbf{M}}{\partial z_i}
  =
  d_i'(\mathbf{z})\,\be_i\be_i^\top,
\end{align}
where $\be_i$ is the $i$th canonical basis vector. Using the identity
\begin{align}
  \frac{\partial}{\partial z_i}\tr\!\big(\mathbf{M}^{-1}\mathbf{A}\big)
  =
  -\tr\!\left(
    \mathbf{M}^{-1}
    \frac{\partial \mathbf{M}}{\partial z_i}
    \mathbf{M}^{-1}\mathbf{A}
  \right),
\end{align}
it follows that
\begin{align}
  \frac{\partial \phi}{\partial z_i}
  &=
  d_i'(\mathbf{z})\,\be_i^\top \mathbf{C}(\mathbf{z}) \be_i
  \nonumber\\
  &=
  -2\sigma^2 z_i^{-3}\,[\mathbf{C}(\mathbf{z})]_{ii},
  \label{eq:sampling_gradient}
\end{align}
for $i=1,\dots,MW$, where
\begin{align}
  \mathbf{C}(\mathbf{z})
  &=
  \mathbf{M}(\mathbf{z})^{-1}
  \bA
  \mathbf{M}(\mathbf{z})^{-1}.
\end{align}
Given the gradient, one projected-gradient iteration is
\begin{align}
  \bar{\mathbf{z}}^{(j+1)}
  &=
  \mathbf{z}^{(j)}
  -
  \alpha_j \nabla \phi\!\left(\mathbf{z}^{(j)}\right),
  \\
  \mathbf{z}^{(j+1)}
  &=
  \Pi_{\mathcal{H}_K}\!\left(\bar{\mathbf{z}}^{(j+1)}\right),
  \label{eq:sampling_projected_gradient}
\end{align}
where $\alpha_j>0$ is chosen by backtracking line search.
The feasible set $\mathcal{H}_K$ in \eqref{eq:hypersimplex_sampling} is the
capped simplex
\begin{align}
  \mathcal{H}_K
  =
  \left\{
    \mathbf{z}\in\RR^{MW}:
    \mathbf{1}^\top \mathbf{z}=K,\;
    \varepsilon \le z_i \le 1
  \right\}.
\end{align}
Its Euclidean projection has the thresholding form
\begin{align}
  \big[\Pi_{\mathcal{H}_K}(\mathbf{v})\big]_i
  =
  \min\!\big\{1,\max\{\varepsilon,\,v_i-\tau\}\big\},
  \label{eq:sampling_projection}
\end{align}
where the scalar $\tau$ is selected such that
\begin{align}
  \sum_{i=1}^{MW}
  \min\!\big\{1,\max\{\varepsilon,\,v_i-\tau\}\big\}
  = K.
\end{align}
This projection can be computed efficiently by bisection on $\tau$.

\bibliographystyle{IEEEbib}
\bibliography{export}

@article{juhlin2023optimal,
  title={Optimal sensor placement for localizing structured signal sources},
  author={Juhlin, M. and Jakobsson, A.},
  journal={Signal Processing},
  volume={202},
  pages={108679},
  year={2023},
  publisher={Elsevier}
}

@book{hayes1996statistical,
  title={Statistical digital signal processing and modeling},
  author={Hayes, Monson H},
  address = {New York},
  year={1996},
  publisher={John Wiley \& Sons}
}

@article{nishida2022region,
  title={Region-restricted sensor placement based on {G}aussian process for sound field estimation},
  author={Nishida, T. and Ueno, N. and Koyama, S. and Saruwatari, H.},
  journal={IEEE Trans. Sig. Process.},
  volume={70},
  pages={1718--1733},
  year={2022},
  publisher={IEEE}
}

@inproceedings{ariga2020mutual,
  title={Mutual-information-based sensor placement for spatial sound field recording},
  author={Ariga, K. and Nishida, T. and Koyama, S. and Ueno, N. and Saruwatari, H.},
  booktitle={IEEE Int. Conf. Acoust., Speech, Signal Process.},
  pages={166--170},
  year={2020},
  month = {05},
  address = {Barcelona, Spain},
  organization={IEEE}
}

@article{garcia1997generation,
  title={Generation of zones of quiet using a virtual microphone arrangement},
  author={Garcia-Bonito, J. and Elliott, S. J. and Boucher, C. C.},
  journal={J. Acoust. Soc. Amer.},
  volume={101},
  number={6},
  pages={3498--3516},
  year={1997},
  publisher={Acoustical Society of America}
}

@article{zhang2018active,
  title={Active noise control over space: A wave domain approach},
  author={Zhang, J. and Abhayapala, T. D. and Zhang, W. and Samarasinghe, P. N. and Jiang, S.},
  journal={{IEEE} Trans. Audio, Speech, Lang. Process.},
  volume={26},
  number={4},
  pages={774--786},
  year={2018},
  publisher={IEEE}
}

@inproceedings{spors2007approach,
  title={An approach to massive multichannel broadband feedforward active noise control using wave-domain adaptive filtering},
  author={Spors, S. and Buchner, H.},
  booktitle={IEEE workshop on applications of signal process. to audio and acoust.},
  pages={171--174},
  year={2007},
  address={New Paltz, New York},
  month= {10},
}

@book{Abramowitz1972,
  title={Abramowitz and Stegun: Handbook of Mathematical Functions},
  author={M. Abramowitz and I. A. Stegun},
  year={1972},
  publisher={US Department of Commerce},
  address = {Washington, DC}
}

@article{sundstrom2025bayesian,
  title={Bayesian Sound Field Reconstruction Using Partial Boundary Information},
  author={Sundstr{\"o}m, D. and Elvander, F. and Jakobsson, A.},
  journal={IEEE Trans. Audio, Speech, Lang. Process.},
  year={2025},
  volume = {33},
  pages = {4620--4631}
}

@inproceedings{brunnstrom2025spatial,
  title={Spatial covariance estimation for sound field reproduction using kernel ridge regression},
  author={Brunnstr{\"o}m, J. and M{\o}ller, M. B. and {\O}stergaard, J. and van Waterschoot, T. and Moonen, M. and Elvander, F.},
  booktitle={IEEE European Sig. Process. Conf.},
  pages={86--90},
  year={2025},
  month={09},
  address={Palermo, Italy}
}

@article{amengual2021optimizations,
  title={Optimizations of the spatial decomposition method for binaural reproduction},
  author={Gari, S. V. A. and Arend, J. M. and Calamia, P. T. and Robinson, P. W.},
  journal={J. Audio Eng. Soc.},
  volume={68},
  number={12},
  pages={959--976},
  year={2021},
  publisher={Audio Engineering Society}
}

@article{pintelon2002frequency,
  title={Frequency domain system identification using arbitrary signals},
  author={Pintelon, R. and Schoukens, J. and Vandersteen, G.},
  journal={IEEE Trans. on Automatic Control},
  volume={42},
  number={12},
  pages={1717--1720},
  year={2002},
  publisher={IEEE}
}

@inproceedings{schoukens2004time,
  title={Time domain identification, frequency domain identification. Equivalencies! Differences?},
  author={Schoukens, J. and Pintelon, R. and Rolain, Y.},
  booktitle={IEEE Amer. Control Conf.},
  volume={1},
  pages={661--666},
  year={2004},
  address = {Boston, MA, USA}
}

@inproceedings{ljung2004state,
  title={State of the art in linear system identification: Time and frequency domain methods},
  author={Ljung, L.},
  booktitle={IEEE Amer. Control Conf.},
  volume={1},
  pages={650--660},
  year={2004},
  address = {Boston, MA, USA}
}

@article{hahmann2022convolutional,
  title={A convolutional plane wave model for sound field reconstruction},
  author={Hahmann, M. and Fernandez-Grande, E.},
  journal={J. Acoust. Soc. Amer.},
  volume={152},
  number={5},
  pages={3059--3068},
  year={2022},
  publisher={AIP Publishing}
}

@book{jacobsen2013fundamentals,
  title={Fundamentals of general linear acoustics},
  author={Jacobsen, F. and Juhl, P. M.},
  year={2013},
  address={New York},
  publisher={John Wiley \& Sons}
}

@inproceedings{sundstrom2023recursive,
  title={Recursive spatial covariance estimation with sparse priors for sound field interpolation},
  author={Sundstr{\"o}m, D. and Lindstr{\"o}m, J. and Jakobsson, A.},
  booktitle={IEEE Statistical Signal Processing Workshop},
  pages={517--521},
  year={2023},
  address = {Hanoi, Vietnam},
  month = {06}
}

@article{samarasinghe2015efficient,
  title={An efficient parameterization of the room transfer function},
  author={Samarasinghe, P. and Abhayapala, T. and Poletti, M. and Betlehem, T.},
  journal={IEEE Trans. Audio, Speech, Lang. Process.},
  volume={23},
  number={12},
  pages={2217--2227},
  year={2015},
  publisher={IEEE}
}

@article{badeau2024statistical,
  title={Statistical wave field theory},
  author={Badeau, R.},
  journal={J. Acoust. Soc. Amer.},
  volume={156},
  number={1},
  pages={573--599},
  year={2024},
  publisher={AIP Publishing}
}

@article{koopmann1989method,
  title={A method for computing acoustic fields based on the principle of wave superposition},
  author={Koopmann, G. H. and Song, L. and Fahnline, J. B.},
  journal={J. Acoust. Soc. Amer.},
  volume={86},
  number={6},
  pages={2433--2438},
  year={1989},
  publisher={Acoustical Society of America}
}

@article{mignot2013room,
  title={Room reverberation reconstruction: Interpolation of the early part using compressed sensing},
  author={Mignot, R. and Daudet, L. and Ollivier, F.},
  journal={IEEE Trans. Audio, Speech, Lang. Process.},
  volume={21},
  number={11},
  pages={2301--2312},
  year={2013},
  publisher={IEEE}
}

@article{johnson1998equivalent,
  title={An equivalent source technique for calculating the sound field inside an enclosure containing scattering objects},
  author={Johnson, M. E. and Elliott, S. J. and Baek, K. H. and Garcia-Bonito, J.},
  journal={J. Acoust. Soc. Amer.},
  volume={104},
  number={3},
  pages={1221--1231},
  year={1998},
  publisher={Acoustical Society of America}
}

@article{tervo2013spatial,
  title={Spatial decomposition method for room impulse responses},
  author={Tervo, S. and P{\"a}tynen, J. and Kuusinen, A. and Lokki, T.},
  journal={J. Audio Eng. Soc.},
  volume={61},
  number={1/2},
  pages={17--28},
  year={2013},
  publisher={Audio Engineering Society}
}

@article{olivieri2024physics,
  title={Physics-informed neural network for volumetric sound field reconstruction of speech signals},
  author={Olivieri, M. and Karakonstantis, X. and Pezzoli, M. and Antonacci, F. and Sarti, A. and Fernandez-Grande, E.},
  journal={EURASIP J. Audio, Speech, Music Process.},
  volume={2024},
  number={1},
  pages={42},
  year={2024},
  publisher={Springer}
}

@article{brunnstrom2025time,
  title={Time-domain sound field estimation using kernel ridge regression},
  author={Brunnstr{\"o}m, J. and M{\o}ller, M. B. and {\O}stergaard, J. and Koyama, S. and van Waterschoot, T. and Moonen, M.},
  journal={IEEE Trans. Audio, Speech, Lang. Process.},
  year={2026},
  volume = {34},
  pages = {1243--1258},
  publisher={IEEE}
}

@article{caviedes2023spatio,
  title={Spatio-temporal {B}ayesian regression for room impulse response reconstruction with spherical waves},
  author={Caviedes-Nozal, D. and Fernandez-Grande, E.},
  journal={IEEE Trans. Audio, Speech, Lang. Process.},
  volume={31},
  pages={3263--3277},
  year={2023},
}

@inproceedings{ueno2018kernel,
  title={Kernel ridge regression with constraint of {H}elmholtz equation for sound field interpolation},
  author={Ueno, N. and Koyama, S. and Saruwatari, H.},
  booktitle={IEEE Int. Workshop Acoust. Signal Enhancement},
  pages={1--440},
  year={2018},
  month = {09},
  address = {Tokyo, Japan}
  
}

@article{cook1955measurement,
  title={Measurement of correlation coefficients in reverberant sound fields},
  author={Cook, R. K. and Waterhouse, R. V. and Berendt, R. D. and Edelman, S. and Thompson Jr, M. C.},
  journal={J. Acoust. Soc. Amer.},
  volume={27},
  number={6},
  pages={1072--1077},
  year={1955},
  publisher={Acoustical Society of America}
}

@article{gonzalez2010measurement,
  title={Measurement of areas on a sphere using Fibonacci and latitude--longitude lattices},
  author={Gonz{\'a}lez, {\'A}.},
  journal={Mathematical geosciences},
  volume={42},
  pages={49--64},
  year={2010},
  publisher={Springer}
}

@inproceedings{brunnstrom2022variable,
  title={Variable span trade-off filter for sound zone control with kernel interpolation weighting},
  author={Brunnstr{\"o}m, J. and Koyama, S. and Moonen, M.},
  booktitle={{IEEE} Int. Conf. Acoust., Speech, Signal Process.},
  pages={1071--1075},
  year={2022},
  month={05},
  address= {Singapore},
}

@article{karakonstantis2024room,
  title={Room impulse response reconstruction with physics-informed deep learning},
  author={Karakonstantis, X. and Caviedes-Nozal, D. and Richard, A. and Fernandez-Grande, E.},
  journal={J. Acoust. Soc. Amer.},
  volume={155},
  number={2},
  pages={1048--1059},
  year={2024},
  publisher={AIP Publishing}
}

@inproceedings{sundstrom2024sound,
  title={Sound field estimation using deep kernel learning regularized by the wave equation},
  author={Sundstr{\"o}m, D. and Koyama, S. and Jakobsson, A.},
  booktitle={IEEE Int. Workshop Acoust. Signal Enhancement},
  pages={319--323},
  year={2024},
  month = {09},
  address = {Aalborg, Denmark}
}

@article{ueno2017sound,
  title={Sound field recording using distributed microphones based on harmonic analysis of infinite order},
  author={Ueno, N. and Koyama, S. and Saruwatari, H.},
  journal={{IEEE} Signal Process. Lett.},
  volume={25},
  number={1},
  pages={135--139},
  year={2017}
}

@inproceedings{horiuchi2021kernel,
  title={Kernel learning for sound field estimation with l1 and l2 regularizations},
  author={Horiuchi, R. and Koyama, S. and Ribeiro, J. G. C. and Ueno, N. and Saruwatari, H.},
  booktitle={{IEEE} Int. Workshop Appl. Signal Process. Audio Acoust.},
  pages={261--265},
  year={2021},
  address = "New Paltz, NY, USA"
}

@book{williams1999fourier,
  title={Fourier acoustics: sound radiation and nearfield acoustical holography},
  author={Williams, E. G.},
  year={1999},
  address={New York},
  publisher={Academic press}
}

@article{habets2006room,
  title={Room impulse response generator},
  author={Habets, E.A.P.},
  journal={Technische Universiteit Eindhoven, Tech. Rep},
  volume={2},
  number={2.4},
  pages={1},
  year={2006}
}

@inproceedings{fernandez2021reconstruction,
  title={Reconstruction of room impulse responses over extended domains for navigable sound field reproduction},
  author={Fernandez-Grande, E. and Caviedes-Nozal, D. and Hahmann, M. and Karakonstantis, X. and Verburg, S. A.},
  booktitle={IEEE Immersive and 3D Audio: from Architecture to Automotive},
  pages={1--8},
  year={2021},
  address = {Bologna, Italy}
}

@article{ribeiro2024sound,
  title={Sound field estimation based on physics-constrained kernel interpolation adapted to environment},
  author={Ribeiro, J. G. C. and Koyama, S. and Horiuchi, R. and Saruwatari, H.},
  journal={IEEE Trans. Audio, Speech, Lang. Process.},
  year={2024},
  volume={32},
  pages={4369-4383},
  publisher={IEEE}
}

@book{wahba1990spline,
  title={Spline models for observational data},
  author={Wahba, G.},
  year={1990},
  publisher={SIAM}
}

@book{rasmussen2005gaussian,
    author = {Rasmussen, C. E. and Williams, C. K. I.},
    title = {Gaussian Processes for Machine Learning},
    publisher = {The MIT Press},
    year = {2005},
    month = {11},
    isbn = {9780262256834},
    doi = {10.7551/mitpress/3206.001.0001},
    url = {https://doi.org/10.7551/mitpress/3206.001.0001}
}

@article{caviedes2021gaussian,
  title={Gaussian processes for sound field reconstruction},
  author={Caviedes-Nozal, D. and Riis, N.A.B. and Heuchel, F.M. and Brunskog, J. and Gerstoft, P. and Fernandez-Grande, E.},
  journal={J. Acoust. Soc. Amer.},
  volume={149},
  number={2},
  pages={1107--1119},
  year={2021},
  publisher={AIP Publishing}
}

@article{verburg2018reconstruction,
  title={Reconstruction of the sound field in a room using compressive sensing},
  author={Verburg, S. A. and Fernandez-Grande, E.},
  journal={J. Acoust. Soc. Amer.},
  volume={143},
  number={6},
  pages={3770--3779},
  year={2018}
}

@article{verburg2024optimal,
  title={Optimal sensor placement for the spatial reconstruction of sound fields},
  author={Verburg, S. A. and Elvander, F. and van Waterschoot, T. and Fernandez-Grande, E.},
  journal={EURASIP J. Audio, Speech, Music Process.},
  volume={2024},
  number={1},
  pages={41},
  year={2024},
  publisher={Springer}
}

@article{antonello2017room,
  title={Room impulse response interpolation using a sparse spatio-temporal representation of the sound field},
  author={Antonello, N. and De Sena, E. and Moonen, M. and Naylor, P. A and Van Waterschoot, T.},
  journal={{IEEE} Trans. Audio, Speech, Lang. Process.},
  volume={25},
  number={10},
  pages={1929--1941},
  year={2017},
  publisher={IEEE}
}

@ARTICLE{Sundstrom2023,
  author={Sundström, D. and Elvander, F. and Jakobsson, A.},
  journal={IEEE Trans. Signal Process.}, 
  title={Optimal Transport Based Impulse Response Interpolation in the Presence of Calibration Errors}, 
  year={2024},
  pages={1-12},
  doi={10.1109/TSP.2024.3372249}}

@article{Koyama2021,
   author = {S. Koyama and J. Brunnström and H. Ito and N. Ueno and H. Saruwatari},
   doi = {10.1109/TASLP.2021.3107983},
   issn = {23299304},
   journal = {IEEE Trans. Audio Speech and Lang. Process.},
   title = {Spatial Active Noise Control Based on Kernel Interpolation of Sound Field},
   pages = {3052-3063}, 
   volume = {29},
   year = {2021},
}

@article{Naylor2011,
   author = {P. A. Naylor and N. D. Gaubitch and E. Cross},
   doi = {10.3397/1.3532780},
   issn = {07362501},
   issue = {2},
   journal = {Noise Control Engineering Journal},
   title = {Speech Dereverberation},
   volume = {59},
   year = {2011},
}

\end{document}